\documentclass{lmcs} %
\pdfoutput=1

\usepackage{lastpage}
\lmcsdoi{18}{4}{2}
\lmcsheading{}{\pageref{LastPage}}{}{}%
{Jun.~20,~2021}{Oct.~21,~2022}{}


\usepackage{hyperref}
\usepackage{mathrsfs}%
\usepackage[utf8]{inputenc}
\usepackage{amssymb}
\usepackage[capitalize,nameinlink]{cleveref}
\usepackage{tikz}
\usepackage{complexity-mod}

\usepackage[algoruled,vlined,english]{algorithm2e}

\usepackage{appendix}
\usepackage{graphicx}
\usetikzlibrary{shapes,positioning,fit,calc,decorations.text,decorations.pathmorphing,backgrounds}
\usetikzlibrary{automata}
\usetikzlibrary{arrows,decorations.pathreplacing}

\usepackage[shortlabels]{enumitem}

\usepackage{stmaryrd}
\usepackage[algoruled,vlined,english]{algorithm2e}

\usepackage{graphicx}

\usepackage[capitalize]{cleveref}

\usepackage{comment}

\usepackage{listings}

  \lstdefinestyle{tinyc}{
    basicstyle=\scriptsize\ttfamily,
    keywordstyle=\color{blue}
  }

  \lstdefinestyle{normalc}{
    basicstyle=\ttfamily,
    numbers=none,
    keywordstyle=\color{blue}
  }

  \lstdefinestyle{inlinec}{
    basicstyle=\ttfamily
  }

\newtheorem*{remark}{Remark}

\Crefname{thm}{Theorem}{Theorems}
\Crefname{claim}{Claim}{Claims}
\Crefname{defn}{Definition}{Definitions}
\Crefname{lem}{Lemma}{Lemmas}
\Crefname{obs}{Observation}{Observations}
\Crefname{prop}{Proposition}{Propositions}
\Crefname{cor}{Corollary}{Corollaries}
\Crefname{prob}{Problem}{Problems}
\Crefname{prem}{Remark}{Remarks}

\def\@envspa{\hspace{0.3em}}
\def\@sa{\hspace{-0.2em}}
\def\@sb{\hspace{0.5em}}
\def\@sc{\hspace{-0.1em}}

{\bfseries\upshape}{\itshape}
{\bfseries\upshape}{\itshape}

\def\set#1{{\{ #1 \}}}

\def\multi#1{{\llbracket #1 \rrbracket}}
\def\nats{{\mathbb{N}}}

\def\parikh{{\mathsf{Parikh}}}
\def\pre{\mathit{pre}}

\def\mmap{\mathbf{m}}

\def\card#1{\lvert {#1} \rvert}

\def\TWOEXPSPACE{\mathsf{2EXPSPACE}}

\newcommand{\proj}{\mathrm{Proj}}

\newcommand{\dcl}[1]{\mathord{\downarrow}{#1}}

\newcommand{\multiset}[1]{{\mathbb{M}[ #1 ]}}

\def\ap{\mathfrak{P}}

\def\prod{\mathcal{P}}

\def\rs{\mathscr{S}}

\def\PDA{\mathsf{PDA}}
\def\hors{\mathsf{HORS}}

\def\nonterm{\mathcal{N}}
\def\rewrite{\mathcal{R}}
\def \appterms{\tilde{\Theta}}
\DeclareMathOperator{\ord}{ord}

\newcommand{\Var}{\mathrm{Var}}
\newcommand{\otype}{\mathtt{o}} 

\newcommand{\ltr}[1]{\mathtt{#1}}

\def \tree{\mathcal{T}}
\def \paths{\mathsf{Paths}}

\def \e{\mathsf{e}}
\def \br{\ltr{br}}

\def \prefix{\leq_{\mathsf{pre}}}
\def \sprefix{<_{\mathsf{pre}}}

\DeclareMathOperator{\dom}{dom}
\DeclareMathOperator{\arity}{arity}
\DeclareMathOperator{\lub}{lub}

\def \langof{\mathcal{L}}
\def \langpath{\mathcal{L}_\mathsf{p}}
\def \langword{\mathcal{L}_\mathsf{w}}

\def \langtree{\mathcal{L}_\mathsf{t}}

\def \AP{\mathsf{AP}}

\def \api{\mathfrak{P}_\mathsf{i}}
\def \apdown{\dcl{\mathfrak{P}}}
\def \Sigmai{\Sigma_{\mathsf{i}}}
\def \langhors{\mathcal{H}}

\def \subword{\sqsubseteq}

\newcommand{\fC}{\mathfrak{C}} 
\newcommand{\Runs}[1]{\mathsf{Runs}(#1)}
\newcommand{\Preruns}[1]{\mathsf{Preruns}(#1)}
\def \leqruns{\trianglelefteq}
\def\nmap{\mathbf{n}}

\newcommand{\cC}{\mathcal{C}}
\newcommand{\cA}{\mathcal{A}}
\newcommand{\cT}{\mathcal{T}}
\newcommand{\calL}{\mathcal{L}}

\newcommand{\rev}[1]{#1^{\mathsf{rev}}}
\DeclareDocumentCommand{\langof}{O{} m}{%
  \mathsf{L}_{#1}(#2)%
  }
\DeclareDocumentCommand{\autstep}{O{}}{%
        \xrightarrow{#1}%
        }

\DeclareDocumentCommand{\autsteps}{O{}}{%
        \xRightarrow{#1}%
      }
\DeclareDocumentCommand{\autsteph}{O{} m}{%
        \xrightarrow{#1}_{#2}%
        }

\DeclareDocumentCommand{\autstepsh}{O{} m}{%
        \xRightarrow{#1}_{#2}%
      }
\def \process{\mathscr{L}}

\newcommand{\impl}[2]{``\ref{#1}$\Rightarrow$\ref{#2}''}

\mathchardef\mhyphen="2D
\DeclareDocumentCommand{\EXPTIME}{O{}}{%
  \ifthenelse{\equal{#1}{}}{%
    \mathsf{EXPTIME}%
  }{%
    {#1}\mhyphen\mathsf{EXPTIME}%
  }%
}

\begin{document}

\title[Decidability for Asynchronous Shared-Memory Programs]{General
Decidability Results\texorpdfstring{\\}{ }
for Asynchronous Shared-Memory Programs:\texorpdfstring{\\}{ }
 Higher-Order and Beyond}
\titlecomment{An abridged version of this paper appeared in TACAS 2021; missing proofs have been added to create this version.}

\author[R.~Majumdar]{Rupak Majumdar}  %
\address{Max Planck Institute for Software Systems (MPI-SWS), 
Paul-Ehrlich-Stra{\ss}e, Building G26, Kaiserslautern, Germany 67663.} %
\email{\{\texttt{rupak},\texttt{thinniyam},\texttt{georg}\}\texttt{@mpi-sws.org}}  %

\author[R.~S.~Thinniyam]{Ramanathan S. Thinniyam}  %

\author[G.~Zetzsche]{Georg Zetzsche}  %

\begin{abstract}
  \noindent The model of asynchronous programming arises in many contexts, from
  low-level systems software to high-level web programming.
  We take a language-theoretic perspective and show general decidability and undecidability results for asynchronous programs
  that capture all known results as well as show decidability of new and important classes.
  As a main consequence, we show decidability of safety, termination and boundedness verification for \emph{higher-order} 
  asynchronous programs---such as OCaml programs using Lwt---and undecidability of liveness verification already for order-2 asynchronous programs. 
  We show that under mild assumptions, surprisingly, safety and termination verification of asynchronous programs with handlers
  from a language class are decidable \emph{if{}f} emptiness is decidable for the underlying language class.
  Moreover, we show that configuration reachability and liveness (fair termination) verification are
  equivalent, and decidability of these problems implies decidability
  of the well-known ``equal-letters'' problem on languages.
  Our results close the decidability frontier for asynchronous programs.
\end{abstract}

\maketitle
\section{Introduction} %
\label{sec:introduction}

\emph{Asynchronous programming} is a common way to manage concurrent requests in a system.
In this style of programming, rather than waiting for a time-consuming operation to complete,
the programmer can make \emph{asynchronous} procedure calls which are stored in
a \emph{task buffer} pending later execution.
Each asynchronous procedure, or \emph{handler}, is a sequential program.
When run, it can change the \emph{global shared state} of the program, 
make internal synchronous procedure calls, 
and post further instances of handlers to the task buffer. 
A scheduler repeatedly and non-deterministically picks pending handler instances from the task buffer 
and executes their code \emph{atomically} to completion.
Asynchronous programs appear in many domains, such as operating system kernel code, web programming, or
user applications on mobile platforms. 
This style of programming is supported natively or through libraries for most programming environments.
The interleaving of different handlers hides latencies
of long-running operations:
the program can process a different handler while waiting for an external operation to finish.
However, asynchronous scheduling of tasks introduces non-determinism in the system, 
making it difficult to reason about correctness.

An asynchronous program is \emph{finite-data} if 
all program variables range over finite domains.
Finite-data programs are still infinite state transition systems: 
the task buffer can contain an unbounded number of pending instances and 
the sequential machine implementing an individual handler can have unboundedly large state (e.g., 
if the handler is given as a recursive program, the stack can grow unboundedly).
Nevertheless, verification problems for finite-data programs have been shown to be decidable 
for several kinds of handlers~\cite{SenV06,JhalaM07,ChadhaV07,GantyM12}.
Several algorithmic approaches have been studied, which tailor to 
(i)~the kinds of permitted handler programs and (ii)~the properties that are checked.

\paragraph{State of the art} We briefly survey the existing approaches and what is known about the decidability frontier.
The \emph{Parikh approach} applies to (first-order) recursive handler programs. Here, the decision problems for asynchronous programs are reduced to decision problems over Petri nets~\cite{GantyM12}.
The key insight is that since handlers are executed atomically, the order in which a handler posts tasks to the buffer is irrelevant.
Therefore, instead of considering the sequential order of posted tasks along an execution, one can equivalently
consider its Parikh image. %
Thus, when handlers are given as pushdown systems, the behaviors of an asynchronous program can be 
represented by a (polynomial sized) Petri net.
Using the Parikh approach, safety (formulated as reachability of a global state), termination (whether all executions terminate), and boundedness
(whether there is an a priori upper bound on the task buffer) are all decidable for asynchronous programs with recursive handlers, by reduction to corresponding problems on Petri nets~\cite{SenV06,GantyM12}.
Configuration reachability (reachability of a specific global state and task buffer configuration), fair termination
(termination under a fair scheduler), and fair non-starvation (every pending handler instance is eventually executed) 
are also decidable, by separate ad hoc reductions to Petri net reachability~\cite{GantyM12}. 
A ``reverse reduction'' shows that Petri nets can be 
simulated by polynomial-sized asynchronous programs (already with finite-data handlers).

In the \emph{downclosure approach}, one replaces each handler with a finite-data program that is equivalent up to ``losing'' handlers in the task buffer.
Of course, this requires that one can compute equivalent finite-data
programs for a given class of handler programs.
This has been applied to checking safety for recursive handler programs~\cite{AtigBQ2009}.
Finally, a bespoke \emph{rank-based approach} has been applied to checking safety when handlers
can perform restricted higher-order recursion~\cite{ChadhaV07}.

\paragraph{Contribution} Instead of studying individual kinds of handler programs,
we consider asynchronous programs in a general language-theoretic framework.
The class of handler programs is given as a language class $\mathcal{C}$:
An asynchronous program over a language class $\mathcal{C}$ is one where each
handler defines a language from $\mathcal{C}$ over the alphabet of handler names, as well as a transformer over
the global state.
This view leads to general results: we can obtain simple characterizations of which classes of handler programs permit decidability.
For example, we do not need the technical assumptions of computability of equivalent finite-data programs
from the Parikh and the downclosure approach.

Our first result shows that, under a mild language-theoretic assumption, 
safety and termination are decidable if and only if the underlying language class $\mathcal{C}$ has a decidable
emptiness problem.\footnote{
  The ``mild language-theoretic assumption'' is that the class of languages forms an effective full trio:
  it is closed under intersections with regular languages, homomorphisms, and inverse homomorphisms.
  Many language classes studied in formal language theory and verification satisfy these conditions.
}
Similarly, we show that boundedness is decidable iff \emph{finiteness} is decidable for the language class $\mathcal{C}$.
These results are the best possible: decidability of emptiness (resp., finiteness)
is a requirement for safety and termination verification already
for verifying the safety or termination (resp., boundedness) 
of one \emph{sequential} handler call.
As corollaries, we get new decidability results for all these problems for asynchronous programs over \emph{higher-order recursion schemes},
which form the language-theoretic basis for programming in higher-order functional languages such as OCaml~\cite{DBLP:conf/popl/Kobayashi09,DBLP:conf/lics/Ong15},
as well as other language classes (lossy channel languages, Petri net languages, etc.).

Second, we show that configuration reachability, fair termination, and fair starvation are mutually reducible; thus, decidability
of any one of them implies decidability of all of them.
We also show decidability of these problems implies the decidability
of a well-known combinatorial problem on languages:
given a language over the alphabet $\set{\ltr{a},\ltr{b}}$, decide if
it contains a word with an equal number of $\ltr{a}$s and $\ltr{b}$s.
Viewed contrapositively, we conclude that all these decision problems are undecidable already for
asynchronous programs over order-2 pushdown languages, since the equal-letters problem is undecidable for this class.

Together, our results ``close'' the decidability frontier for asynchronous programs, by demonstrating 
reducibilities between decision problems heretofore studied separately and connecting
decision problems on asynchronous programs with decision problems on the underlying language classes of their handlers.

While our algorithms do not assume that downclosures are effectively computable, we use downclosures to prove their correctness.
We show that the safety, termination, and boundedness problems are invariant under taking downclosures of
runs; this corresponds to taking downclosures of the languages of handlers.

The observation that safety, termination, and boundedness depend only on the downclosure 
suggests a possible route to implementation.
If there is an effective procedure to compute the downclosure for class $\mathcal{C}$,
then a direct verification algorithm would replace all handlers by their (regular) downclosures, 
and invoke existing decision procedures for this class.
Thus, we get a direct algorithm based on downclosure constructions for higher-order recursion schemes,
using the string of celebrated recent results on effectively computing the downclosure of \emph{word schemes}~\cite{Zetzsche2015b,HagueKO16,ClementePSW16}.

We find our general decidability result for asynchronous programs to be surprising.
Already for regular languages, the complexity of safety verification jumps from $\NL$
(NFA emptiness) to $\EXPSPACE$ (Petri net coverability): asynchronous programs
are far more expressive than individual handler languages.
It is therefore unexpected that safety and termination 
remain decidable whenever they are decidable for individual handler languages.

\section{Preliminaries} %
\label{sec:preliminaries}

\paragraph{Basic Definitions}

We assume familiarity with basic definitions of automata theory (see,
e.g.,~\cite{DBLP:books/daglib/0016921,DBLP:books/daglib/0086373}).
If $w$ is a word over $\Sigma$ and $\Gamma\subseteq\Sigma$ is a subalphabet, then the \emph{projection of \(w\) onto $\Gamma$}, denoted
$\proj_{\Gamma}(w)$, is the word obtained by erasing from \(w\) each symbol
that does not belong to \(\Gamma\).  For a language $L\subseteq\Sigma^*$, define
\(\proj_{\Gamma}(L)=\set{\proj_{\Gamma}(w)\mid w \in L}\).
The \textit{subword} order $\subword$ on $\Sigma^*$ is
defined as $w \subword w'$ for $w,w' \in \Sigma^*$ if $w$ can be obtained from $w'$
by deleting some letters from $w'$. For example,
$abba \subword bababa$ but $abba \not \subword baaba$. 
The \textit{downclosure} $\dcl{w}$ with respect to the
subword order, of a word $w\in\Sigma^*$, is defined as $\dcl{w} :=
\{w'\in\Sigma^* \mid w' \subword w\}$. The downclosure $\dcl{L}$ of a language
$L\subseteq\Sigma^*$ is given by $\dcl{L}:=\{ w'\in\Sigma^* \mid \exists w \in L \colon w' \subword w\}$.
Recall that the downclosure $\dcl{L}$ of any language $L$ is a regular
language~\cite{haines1969free}.

A \emph{multiset} $\mmap\colon \Sigma\rightarrow\nats$ over $\Sigma$ maps each
symbol of $\Sigma$ to a natural number.
Let $\multiset{\Sigma}$ be the set of all multisets over $\Sigma$.
We treat sets as a special case of multisets 
where each element is mapped onto $0$ or $1$.
For example, we write
$\mmap=\multi{\ltr{a},\ltr{a},\ltr{c}}$ for the multiset
$\mmap\in\multiset{\set{\ltr{a},\ltr{b},\ltr{c},\ltr{d}}}$ with
$\mmap(\ltr{a})=2$,
$\mmap(\ltr{b})=\mmap(\ltr{d})=0$, and $\mmap(\ltr{c})=1$.  We also use the notation $\card{\mmap}=\sum_{\sigma\in\Sigma}\mmap(\sigma)$. The Parikh image $\parikh(w) \in \multiset{\Sigma}$ of a word $w \in \Sigma^*$ is the multiset such that $\parikh(w)(\ltr{a})$ is the number of times $\ltr{a}$ occurs in $w$. 

Given two multisets $\mmap,\mmap'\in\multiset{\Sigma}$ we define the multiset $\mmap\oplus
\mmap'\in\multiset{\Sigma}$ for which, for all
$\ltr{a}\in\Sigma$,
we have $(\mmap\oplus \mmap')(\ltr{a})=\mmap(\ltr{a})+\mmap'(\ltr{a})$.
We also define the natural order
$\preceq$ on $\multiset{\Sigma}$ as follows: $\mmap\preceq\mmap'$ if{}f there
exists $\mmap^{\Delta}\in\multiset{\Sigma}$ such that
$\mmap\oplus\mmap^{\Delta}=\mmap'$. We also define $\mmap' \ominus
\mmap$ for $\mmap \preceq \mmap'$ analogously: for all
$\ltr{a}\in\Sigma$,
we have $(\mmap\ominus \mmap')(\ltr{a})=\mmap(\ltr{a})-\mmap'(
\ltr{a})$.
For $\Gamma\subseteq \Sigma$ we regard $\mmap\in\multiset{\Gamma}$ as a
multiset of $\multiset{\Sigma}$ where undefined values are sent to $0$.

\paragraph{Language Classes and Full Trios} %
\label{sub:full_trios}

A \emph{language class} is a collection of languages, together with
some finite representation. Examples are the regular (e.g. represented by finite automata) or the context-free languages (e.g. represented by pushdown automata ($\PDA$)).
A relatively weak and reasonable assumption on a language class is 
that it is a \emph{full trio}, that is, it is closed 
under each of the following operations:
    taking intersection with a regular language,
    taking homomorphic images, and 
    taking inverse homomorphic images. 
Equivalently, a language class is a full trio iff it is closed under
\textit{rational transductions}~\cite{berstel2013transductions} as we explain below.
  We assume that all full trios $\cC$ considered in this paper are 
  \textit{effective}: Given a language $L$ from $\cC$, a 
        regular language $R$, and a homomorphism $h$, we can compute
        a representation of the languages $L\cap R$, $h(L)$, and $h^{-1}(L)$ in $\cC$.

Many classes of languages studied in formal language theory form effective full trios.
Examples include the regular and the context-free languages~\cite{DBLP:books/daglib/0016921}, 
the indexed languages~\cite{DBLP:journals/jacm/Aho68,DammGoerdt1986}, 
the languages of higher-order pushdown automata~\cite{maslov1974hierarchy}, higher-order
recursion schemes ($\hors$)~\cite{damm1982io,DBLP:conf/lics/HagueMOS08}, 
Petri nets (both coverability and reachability languages)~\cite{Greibach1978,Jantzen1979}, and lossy channel systems (see \cref{sub:gen_safety_and_termination}).
(While $\hors$ are usually viewed as
representing a tree or collection of trees, one can also view
them as representing a word language, as we explain in Section~\ref{sec:higher-order}.)

Informally, a language class defined by non-deterministic
devices with a finite-state control that allows
$\varepsilon$-transitions and imposes no restriction between input
letter and performed configuration changes (such as non-deterministic pushdown automata)
is always a full trio: The three
operations above can be realized by simple modifications of the finite-state control.
The deterministic context-free languages are a class that is
\emph{not} a full trio.

  An \textit{asynchronous transducer} $\cT$ is a tuple $\cT=(
  Q,\Gamma,\Sigma,E,q_0,F)$
  with a finite set of states $Q$, finite output alphabet $\Gamma$,
  finite input 
  alphabet $\Sigma$, a finite set of edges $E \subseteq Q \times 
  \Gamma^*
  \times \Sigma^* \times Q$, initial state $q_0 \in Q$ and set of final 
  states 
  $F \subseteq Q$. We write $p \autstep[v|u] q$ if $(p,v,u,q) \in 
  E$ and the machine reads $u$ in state $p$, outputs $v$ and moves to 
  state $q$. We also write $p \autstep[w|w'] q$ if there are 
  states $q_0,q_1,\dots,q_n$ and words $u_1,u_2, \dots, u_n, v_1,v_2,
  \dots, v_n$ such that $p=q_0$, $q=q_n$, $w'=u_1u_2\cdots u_n$,
  $w=v_1v_2\cdots v_n$ and $q_{i-1} \autstep[v_i|u_i]
  q_{i}$ for all $0 \leq i \leq n$. 

  The \textit{transduction} $T \subseteq \Gamma^* \times 
  \Sigma^*$ generated by the transducer $\cT$ is the set of tuples 
  $(v,u) \in \Gamma^* \times \Sigma^*$ such that $q_0 \autstep[v|u]
  q_f$ for some $q_f \in F$. Given a language $L \subseteq \Sigma^*$,
  we define 
  $TL:=\{ v \in \Gamma^* \; | \; \exists u \in L \colon  (v,u) \in T\}$. 
 A transduction $T\subseteq\Gamma^*\times\Sigma^*$ is \emph{rational}
 if it is generated by some asynchronous transducer. 

A language class $\cC$ is a \textit{full trio} if and only if it is 
  closed 
  under each of the following operations:
  \begin{itemize}
    \item intersection with a regular language,
    \item taking homomorphic images, and 
    \item taking inverse homomorphic images. 
  \end{itemize}

It is well known that these two concepts are equivalent:
\begin{thm}[Berstel~\cite{berstel2013transductions}]
   A language class is closed under rational transductions if and only if it is a full trio. 
\end{thm}

\paragraph{Asynchronous Programs: A Language-Theoretic View}
\label{sub:asyncprograms}

We use a language-theoretic model for asynchronous shared-memory programs.

\begin{defi}
  \label{defn:ap}
Let $\cC$ be an (effective) full trio. 
  An \emph{asynchronous program} ($\AP$) over $\cC$ is a 
  tuple $\ap=(D,\Sigma,(L_c)_{c\in\fC},d_0,\mmap_0)$, where 
$D$ is a finite set of \emph{global states},  $\Sigma$ is an alphabet 
of \emph{handler names},
$(L_c)_{c\in\fC}$ is a family of languages over $\Sigma$ from $\cC$, one for each $c \in \fC$
where $\fC=D \times \Sigma \times D$ is the set of \textit{contexts},
$d_0\in D$ is the \emph{initial state}, and $\mmap_0\in \multiset{\Sigma}$ is a multiset 
of \emph{initial pending handler instances}.

A {\em configuration} $(d, \mmap) \in D\times \multiset{\Sigma}$ of \( \ap\) 
consists of a global state $d$ and a multiset $\mmap$ of pending handler instances.  For a 
configuration $c$, we write $c.d$ and $c.\mmap$ for the global state and the 
multiset in the configuration respectively.  The \emph{initial} configuration 
\(c_0\) of \(\ap\) is given by \(c_0.d=d_0\) and \(c_0.\mmap=\mmap_0\). 
The semantics of $\ap$ is a labeled transition system over
the set of configurations, with the
transition relation  
$\xrightarrow{\sigma} \subseteq (D \times 
\multiset{\Sigma}) \times (D \times 
\multiset{\Sigma}) $
 given by 
\begin{gather*}
(d,\mmap \oplus \multi{\sigma}) 
\xrightarrow{\sigma}
(d',\mmap \oplus \mmap') 
\quad \text{ if{}f }\quad
\exists w \in L_{d \sigma d'} \colon \parikh(w)=\mmap' 
\end{gather*}
Moreover, we write $(d,\mmap)\rightarrow(d',\mmap')$ if there exists a $\sigma\in\Sigma$ with $(d,\mmap)\xrightarrow{\sigma}(d',\mmap')$. We use 
$\rightarrow^*$
to denote the 
reflexive transitive closure of the relation $\rightarrow$. A 
configuration $c$ is said to be \emph{reachable} in $\ap$ if 
$(d_0,\mmap_0) 
\rightarrow^*
c$. 
\end{defi}

Intuitively, the set $\Sigma$ of handler names specifies a finite set of
procedures that can be invoked asynchronously. 
The shared state takes values in $D$.
When a handler is called asynchronously, it gets added to a bag of
pending handler calls (the multiset $\mmap$ in a configuration).
The language $L_{d\sigma d'}$ captures the effect of executing an instance
of $\sigma$ starting from the global state $d$, such that on 
termination, the global state is $d'$.
Each word $w\in L_{d\sigma d'}$ captures a possible sequence of handlers posted during the execution.

Suppose the current configuration is $(d, \mmap)$.
A non-deterministic scheduler picks one of the outstanding handlers
$\sigma \in \mmap$ and executes it.
Executing $\sigma$ corresponds to picking one of the
languages $L_{d\sigma d'}$ and some word $w\in L_{d\sigma d'}$.
Upon execution of $\sigma$, the new configuration has global state $d'$ and the new bag of pending
calls is obtained by taking $\mmap$, removing an instance of $\sigma$ from it,
and adding the Parikh image of $w$ to it.
This reflects the current set of pending handler calls---the old ones
(minus an instance of $\sigma$) together with the new ones added by
executing $\sigma$.
Note that a handler is executed atomically; thus, we atomically update
the global state and the effect of executing the handler.

Let us see some examples of asynchronous programs.
It is convenient to present these examples in a programming language syntax,
and to allow each handler to have \emph{internal actions} that perform
local tests and updates to the global state.
As we describe informally below, and formally in \cref{sec:internal-actions},
when $\cC$ is a full trio, internal actions can be ``compiled away'' by
taking an intersection with a regular language of internal actions and
projecting the internal actions away.
Thus, we use our simpler model throughout.

\begin{figure}[t]
\begin{lstlisting}[style=tinyc,name=cfl]
global var turn = ref 0 and x = ref 0;
let rec s1 () = if * then begin post a; s1(); post b end
let rec s2 () = if * then begin post a; s2(); post b end else post b
let a () = if !turn = 0 then begin turn := 1; x := !x + 1 end else post a
let b () = if !turn = 1 then begin turn := 0; x := !x - 1 end else post b

let s3 () = post s3; post s3

global var t = ref 0;
let c () = if !t = 0 then t := 1 else post c
let d () = if !t = 1 then t := 2 else post d
let f () = if !t = 2 then t := 0 else post f

let cc x = post c; x
let dd x = post d; x
let ff x = post f; x
let id x = x
let h g y = cc (g (dd y)) 
let rec produce g x = if * then produce (h g) (ff x) else g x
let s4 () = produce id  () 
\end{lstlisting}
\caption{Examples of asynchronous programs}
\label{fig:ex}
\end{figure}

\paragraph{Examples} %
\label{sub:formal_model}

Figure~\ref{fig:ex} shows some simple examples of asynchronous programs in
an OCaml-like syntax.
Consider first the asynchronous program in lines 1--5.
The alphabet of handlers is \texttt{s1}, \texttt{s2}, \texttt{a}, and \texttt{b}.
The global states correspond to possible valuations to the global variables
\texttt{turn} and \texttt{x}; assuming \texttt{turn} is a Boolean and \texttt{x} takes 
values in $\mathbb{N}$,
we have that $D = \set{0,1} \times \set{0,1,\omega}$, where $\omega$
abstracts all values other than $\set{0,1}$.
Since $\mathtt{s1}$ and $\texttt{s2}$ do not touch any
variables, for $d,d'\in D$, we have
$L_{d,\mathtt{s1},d} = \set{\mathtt{a}^n \mathtt{b}^n \mid n\geq 0}$,
$L_{d,\mathtt{s2},d} = \set{\mathtt{a}^{n} \mathtt{b}^{n+1} \mid n\geq
  0}$, and $L_{d,\mathtt{s1},d'}=L_{d,\mathtt{s2},d'}=\emptyset$ if
$d'\ne d$.

For the languages corresponding to \texttt{a} and \texttt{b}, we use syntactic sugar
in the form of \emph{internal actions}; these are local tests and updates to
the global state.
In our example, we have 
$L_{(0,0),\mathtt{a},(1,1)} = \set{\varepsilon}$,
$L_{(1,x),\mathtt{a},(1,x)} = \set{\mathtt{a}}$ for all values of $x$, and similarly for \texttt{b}.
The meaning is that, 
starting from a global state $(0, 0)$, executing the handler
will lead to the global state $(1, 1)$ and no handlers will be posted, 
whereas 
starting from a global state in which \texttt{turn} is $1$,
executing the handler
will keep the global state unchanged but post an instance of \texttt{a}.
Note that all the languages are context-free.

Consider an execution of the program from the initial configuration $((0,0), \multi{\mathtt{s1}})$.
The execution of \texttt{s1} puts $n$ \texttt{a}s and $n$ \texttt{b}s into the bag, for some $n\geq 0$.
The global variable \texttt{turn} is used to ensure that the handlers \texttt{a} and \texttt{b} alternately update \texttt{x}.
When \texttt{turn} is $0$, the handler for \texttt{a} increments \texttt{x} and sets \texttt{turn} to $1$, otherwise it re-posts itself for a future execution.
Likewise, when \texttt{turn} is $1$, the handler for \texttt{b} decrements \texttt{x} and sets \texttt{turn} back to $0$, 
otherwise it re-posts itself for a future execution.
As a result, the variable \texttt{x} never grows beyond $1$.
Thus, the program satisfies the \emph{safety} property that no execution sets \texttt{x} to $\omega$.

It is possible that the execution goes on forever: for example, if \texttt{s1}
posts an \texttt{a} and a \texttt{b}, and thereafter only \texttt{b} is chosen by the scheduler.
This is not an ``interesting'' infinite execution as it is not fair to the pending \texttt{a}.
In the case of a fair scheduler, which eventually always picks an
instance of every pending task,
the program terminates: eventually all the \texttt{a}s and \texttt{b}s are consumed when
they are scheduled in alternation.
However, if instead we started with $\multi{\mathtt{s2}}$, the program will not terminate even under a fair
scheduler: the last remaining \texttt{b} will not be paired and will keep executing and re-posting itself
forever.

Now consider the execution of \texttt{s3}.
It has an infinite fair run, where the scheduler picks an instance of \texttt{s3} at each step.
However, the number of pending instances grows without bound.
We shall study the \emph{boundedness problem}, which checks if the bag can become unbounded along some run.
We also study a stronger notion of fair termination, called \emph{fair non-starvation}, which asks that
every \emph{instance} of a posted handler is executed under any fair scheduler.
The execution of \texttt{s3} is indeed fair, but there can be a specific instance of \texttt{s3} that is never picked:
we say \texttt{s3} can \emph{starve} an instance.

The program in lines 9--20 is \emph{higher-order} ($\mathtt{produce}$ and $\mathtt{h}$ take functions as arguments).
The language of $\mathtt{s4}$ is the set 
$\set{\mathtt{c}^n \mathtt{d}^n \mathtt{f}^n \mid n\geq 0}$,
that is, it posts an equal number of $\mathtt{c}$s, $\mathtt{d}$s, and $\mathtt{f}$s.
It is an indexed language;
we shall see (Section~\ref{sec:higher-order}) how this and other higher-order programs can be represented using higher-order recursion
schemes ($\hors$).
Note that the OCaml types of $\mathtt{produce} : (\otype \rightarrow \otype) \rightarrow \otype \rightarrow \otype$ and
$\mathtt{h} : (\otype \rightarrow \otype) \rightarrow \otype \rightarrow \otype$ are higher-order.

The program is similar to the first: the handlers \texttt{c}, 
\texttt{d}, and \texttt{f} 
execute in ``round robin'' fashion using the global state \texttt{t} to find their turns.
Again, we use internal actions to update the global state for
readability. 
We ask the same decision questions as before: does the program ever reach a specific global
state and does the program have an infinite (fair) run?
We shall see later that safety and termination questions remain decidable, whereas fair termination does not.

\section{Decision Problems on Asynchronous Programs} %
\label{sec:properties_of_asynchronous_programs}

We now describe decision problems on runs of asynchronous programs.

\paragraph{Runs, preruns, and downclosures}
A \textit{prerun} of an $\AP$ $\ap=(D,\Sigma,(L_c)_
  {c\in\fC},d_0,\mmap_0)$ is a finite or infinite sequence $\rho=
  (e_0,\nmap_0),\sigma_1,
  (e_1,\nmap_1),\sigma_2,\ldots$ of alternating elements of tuples $
  (e_i,\nmap_i) \in D 
  \times 
  \multiset{\Sigma}$ and symbols $\sigma_i \in \Sigma$. The set of 
  preruns of $\ap$
  will be denoted 
  $\Preruns{\ap}$. Note that if two asynchronous programs $\ap$
        and $\ap'$ have the same $D$ and 
  $\Sigma$, then $\Preruns{\ap}=\Preruns{\ap'}$. The \emph{length},
        denoted $|\rho|$, of a finite prerun 
  $\rho$ is the number of configurations in $\rho$. The $i^{th}$ configuration of 
  a prerun $\rho$ will be denoted $\rho(i)$. 

  We define an order $\leqruns$ on preruns as follows: For preruns 
  $\rho=
  (e_0,\nmap_0),\sigma_1,
  (e_1,\nmap_1),\sigma_2,\ldots$ and $\rho'=(e_0',\nmap_0'),\sigma_1',
  (e_1',\nmap_1'),\sigma_2',\ldots$, we define $\rho \leqruns \rho'$ if
        $|\rho|=|\rho'|$ and
   $e_i=e_i', \sigma_i=\sigma_i'$ and $\nmap_i 
  \preceq \nmap_i'$ for each $i\geq 0$.  
  The \emph{downclosure} $\dcl{R}$ of a set $R$ of preruns of 
  $\ap$ is defined 
  as
  $\dcl{R}=\{ \rho\in\Preruns{\ap} \mid \exists \rho' \in 
  R. \; \rho \leqruns \rho' \}$.

  A \textit{run} of an $\AP$ $\ap=(D,\Sigma,(L_c)_
  {c\in\fC},d_0,\mmap_0)$ is a prerun $\rho=(d_0,\mmap_0),\sigma_1,
  (d_1,\mmap_1),\sigma_2,\ldots$ starting with the 
  initial 
  configuration $(d_0,\mmap_0)$, where for each $i \geq 0$, we have $
  (d_i,\mmap_i) 
  \xrightarrow{\sigma_{i+1}}
        (d_{i+1},\mmap_{i+1})$. 
    The set of 
    runs of $\ap$ is denoted $\Runs{\ap}$. Finally, $\dcl{\Runs{\ap}}$ is the downward closure of $\Runs{\ap}$ with respect to $\leqruns$. 

An infinite run $c_0\xrightarrow{\sigma_0} c_1\xrightarrow{\sigma_1}\ldots$ is called \emph{fair} if
for all $i\geq 0$, if $\sigma\in c_i.\mmap$ then there is some $j \geq
i$ such that $c_j \xrightarrow{\sigma} c_{j+1}$.
That is, whenever an instance of a handler $\sigma$ is posted, some instance of $\sigma$ is executed later.
Fairness does not guarantee that every specific instance of a handler is executed eventually.
We say that an infinite fair run \emph{starves} a handler $\sigma$ if there exists an index $J \geq 0$ such that for each
$j\geq J$, we have (i) $c_j.\mmap(\sigma) \geq 1$  and (ii) whenever $c_j \xrightarrow{\sigma} c_{j+1}$, we have $c_j.\mmap(\sigma) \geq 2$.
In this case, even if the run is fair, a specific instance of $\sigma$ may never be executed.

Now we give the definitions of the various decision problems.

\begin{defi}[Properties of finite runs]
  \hspace{0pt}
  The {\bf Safety (Global state reachability)} problem asks, given
    an asynchronous program $\ap$ and a global state $d_f\in D$,
    is there a reachable configuration \(c\) such that \(c.d=d_f\)?
    If so, $d_f$ is said to be \emph{reachable} (in \(\ap\)) and \emph{unreachable} otherwise. 
  The {\bf Boundedness (of the task buffer)} problem asks, given
    an asynchronous program $\ap$,
    is there an $N\in\mathbb{N}$ such that 
    for every reachable configuration \(c\), we have $\card{c.\mmap} \leq N$?
    If so, the asynchronous program $\ap$ is \emph{bounded}; otherwise it is
    \emph{unbounded}. 
        The {\bf Configuration reachability} problem asks,
    given an asynchronous program $\ap$ and a configuration \(c\), is
    \(c\) reachable?
\label{def:finite-aaruns}
\end{defi}

\begin{defi}[Properties of infinite runs]
All the following problems take as input an asynchronous program $\ap$.
The {\bf Termination} problem asks if all runs of $\ap$ are finite.
The {\bf Fair Non-termination} problem asks if $\ap$ has some \textit{fair} infinite run.
The {\bf Fair Starvation} problem asks if $\ap$ has some fair run that starves some handler.
\label{def:infinite-aaruns}
\end{defi}

Our main result in this section shows that many properties of an asynchronous
program $\ap$
only depend on the downclosure $\dcl{\Runs{\ap}}$ of the set
$\Runs{\ap}$ of runs of the program $\ap$.
The proof is by induction on the length of runs. 
Please see Appendix \ref{appendix:dcpreservation} for details.
For any $\AP$ $\ap=(D,\Sigma,(L_c)_{c\in\fC},d_0,\mmap_0)$, we
define the $\AP$ $\apdown =(D,
\Sigma,
(\dcl{L_c})_{c \in \fC},
 d_0, \mmap_0)$, where $\dcl{L_c}$ is the downclosure of the language
 $L_c$ under the subword order.

\begin{prop}%
  \label[prop]{dcpreservation}
  Let $\ap=(D,\Sigma,(L_c)_{c\in\fC},d_0,\mmap_0)$ be an asynchronous
  program. Then $\dcl{\Runs{\dcl{\ap}}}=\dcl{\Runs{\ap}}$. In
  particular, the following holds.
(1) For every $d\in D$, $\ap$ can reach $d$ if and only if $\dcl{\ap}$ can reach $d$.
(2) $\ap$ is terminating if and only if $\dcl{\ap}$ is terminating.
(3) $\ap$ is bounded if and only if $\dcl{\ap}$ is bounded.
\end{prop}

Intuitively, safety, termination, and boundedness is preserved when the multiset of
pending handler instances is ``lossy'': posted handlers can get lost.
This corresponds to these handlers never being added to the task buffer.
However, if a run demonstrates reachability of a global state, or non-termination,
or unboundedness, in the lossy version, it corresponds also to a run in the original
problem (and conversely).

In contrast, simple examples show that
configuration reachability, fair termination, and fair non-starvation properties are
not preserved under downclosures.

\tikzset{gadget/.style={->,>=stealth,initial text=,minimum size=7pt,auto,on grid,scale=1,inner sep=1pt,node distance=2.8cm}}
\tikzset{every state/.style={minimum size=15pt,inner sep=1pt,fill=black!10,draw=black!70,thick}}

\section{General Decidability Results}\label{general}
In this section, we characterize those full trios $\cC$ for which
particular problems for asynchronous programs over $\cC$ are
decidable.  
Our decision procedures will use the following theorem, summarizing the results from~\cite{GantyM12}, as a
subprocedure.

\begin{thmC}[\cite{GantyM12}]
\label{thm:GantyM}
Safety, boundedness, configuration reachability, termination, fair non-termination,
and fair non-starvation are decidable for asynchronous programs
over regular languages.
\end{thmC}

\subsection{Safety and termination}
\label{sub:gen_safety_and_termination}
Our first main result concerns the problems of safety and termination.
\begin{thm}\label{general:emptiness}
  Let $\cC$ be a full trio. The following are equivalent:
  \begin{enumerate}[(i)]
  \item\label{general:emptiness:safety} Safety is decidable for asynchronous programs over $\cC$.
  \item\label{general:emptiness:termination} Termination is decidable for asynchronous programs over $\cC$.
  \item\label{general:emptiness:emptiness} Emptiness is decidable for $\cC$.
  \end{enumerate}
\end{thm}
\begin{proof}
We begin with \impl{general:emptiness:safety}{general:emptiness:emptiness}.
Let
$K\subseteq\Sigma^*$ be given. We construct 
\[ \ap=(D,\Sigma,(L_c)_{c\in\fC},d_0,\mmap_0)\]
 such that $\mmap_0=\multi{\sigma}$, $D=
\{d_0,d_1\}$, $L_
{d_0,\sigma,d_1}=K$
and $L_c=\emptyset$ for $c \neq (d_0,\sigma,d_1)$. 
We see that $\ap$ can reach
$d_1$ iff $K$ is non-empty. Next we show
\impl{general:emptiness:termination}
{general:emptiness:emptiness}.
Consider the alphabet $\Gamma=(\Sigma\cup\{\varepsilon\})\times\{0,1\}$ and the homomorphisms
$g\colon\Gamma^*\to\Sigma^*$ and $h\colon\Gamma^*\to\{\sigma\}^*$, where for $x\in\Sigma\cup\{\varepsilon\}$, we have
$g((x,i))=x$ for $i\in\{0,1\}$, $h((x,1))=\sigma$, and $h((x,0))=\varepsilon$. If $R\subseteq\Gamma^*$ is the regular set of words in which exactly
one position belongs to the subalphabet $(\Sigma\cup\{\varepsilon\})\times\{1\}$,
then the language $K':=h(g^{-1}(K)\cap R)$ belongs to $\cC$.
Note that $K'$ is $\emptyset$ or $\{\sigma\}$, depending on whether $K$ is empty or not.
We construct $\ap=(D,\{\sigma\},(L_c)_{c\in\fC},d_0,\mmap_0)$ with $D=\{d_0\}$,
$\mmap_0=\multi{\sigma}$, $L_{d_0,\sigma,d_0}=K'$ and all languages
$L_c=\emptyset$ for $c \neq (d_0,\sigma,d_0)$. Then
$\ap$ is terminating iff $K$ is empty.

To prove \impl{general:emptiness:emptiness}{general:emptiness:safety},
we design an algorithm deciding safety assuming decidability of
emptiness. Given asynchronous program $\ap$ and state $d$ as input,
the algorithm consists of two semi-decision
procedures: one
which searches for a run of $\ap$ reaching the state $d$, and the
second which enumerates regular overapproximations $\ap'$ of $\ap$ and
checks the safety of $\ap'$ using \cref{thm:GantyM}. 
Each $\ap'$ consists of a regular language $A_c$
overapproximating $L_c$ for each context $c$ of $\ap$. We use
decidability of emptiness to
check that $L_c \cap (\Sigma^*\setminus A_c)=\emptyset$ to ensure
that $\ap'$ is indeed an overapproximation.

\begin{algorithm}[t]
  \SetKwBlock{Conc}{run concurrently}{end}
  \SetKwBlock{Begin}{begin}{end}
{\small
  \KwIn{Asynchronous program $\ap=(D,\Sigma,(L_c)_{c\in\fC},d_0,\mmap_0)$ over $\cC$, state $d\in D$}
  \Conc{
    \Begin(\tcc*[f]{find a safe overapproximation}){
      \ForEach{tuple $(A_c)_{c\in\fC}$ of regular languages $A_c\subseteq \Sigma^*$}{
        \If{$L_c\cap(\Sigma^*\setminus A_c)=\emptyset$ for each $c\in\fC$}{
          \If{$\ap'=(D,\Sigma,(A_c)_{c\in\fC},d_0,\mmap_0)$ does not reach $d$}{
            \Return{$d$ is not reachable.}
          }
        }
      }
    }
    \Begin(\tcc*[f]{find a run reaching $d$}){
      \ForEach{prerun $\rho$ of $\ap$}{
        \If{$\rho$ is a run that reaches $d$}{
          \Return{$d$ reachable.}
        } 
      }      
    }
  }
}
  \caption{Checking Safety}\label{algsafety}
\end{algorithm}

Algorithm \ref{algsafety} clearly gives a correct answer if it terminates.
Hence,
we only have to argue that it always does terminate. Of course, if $d$
is reachable, the first semi-decision procedure will terminate. 
In the other case, termination is due to the regularity of downclosures: 
if $d$ is not reachable in $\ap$, then \cref{dcpreservation} tells us
that $\dcl{\ap}$ cannot reach $d$ either.
But $\dcl{\ap}$ is an asynchronous program over regular languages; this means there exists a
safe regular overapproximation and the second semi-decision procedure terminates.

To prove \impl{general:emptiness:emptiness}
{general:emptiness:termination}, we adopt a similar method as for
safety.
Algorithm
\ref{algtermination} for termination consists
of two
semi-decision procedures.  
By standard well-quasi-ordering arguments, an infinite run of an asynchronous
program $\ap$ is witnessed by a
finite self-covering run.
The first semi-decision procedure enumerates finite self-covering runs (trying to show
non-termination).
The second procedure enumerates regular 
asynchronous programs $\ap'$ that overapproximate $\ap$.
As before, to check termination of $\ap'$, it
applies the procedure from Theorem~\ref{thm:GantyM}. 
Clearly, the algorithm's answer is always correct. 
Moreover, it gives an answer for every
input. 
If $\ap$ does not terminate, it will find a self-covering sequence.  
If $\ap$ does terminate, then \cref{dcpreservation} tells
us that $\dcl{\ap}$ is a terminating finite-state overapproximation.
This implies that the second procedure will terminate in that case. \qedhere
\end{proof}

\begin{algorithm}[t]
  \SetKwBlock{Conc}{run concurrently}{end}
  \SetKwBlock{Begin}{begin}{end}
{\small
  \KwIn{Asynchronous program $\ap=(D,\Sigma,(L_c)_{c\in\fC},d_0,\mmap_0)$ over $\cC$}
  \Conc{
    \Begin(\tcc*[f]{find a terminating overapproximation}){
      \ForEach{tuple $(A_c)_{c\in\fC}$ of regular languages $A_c\subseteq \Sigma^*$}{
        \If{$L_c\cap (\Sigma^*\setminus A_c)=\emptyset$ for each $c\in\fC$}{
          \If{$\ap'=(D,\Sigma,(A_c)_{c\in\fC},d_0,\mmap_0)$ terminates}{
            \Return{$\ap$ terminates.}
          }
        }
      }
    }
    \Begin(\tcc*[f]{find a self-covering run}){
      \ForEach{prerun $\rho$ of $\ap$}{
        \If{$\rho$ is a self-covering run}{
          \Return{$\ap$ does not terminate.}
        } 
      }      
    }
  }
}
  \caption{Checking Termination}\label{algtermination}
\end{algorithm}

Let us point out a particular example. 
The class $\calL$ of languages of lossy channel systems
is defined like the class of languages of Well-Structured Transition Systems (WSTS) with upward-closed sets
of accepting configurations as in~\cite{DBLP:journals/acta/GeeraertsRB07}, except that we only consider
lossy channel systems~\cite{DBLP:conf/cav/AbdullaBJ98} instead of
arbitrary WSTS. Then $\calL$
forms a full trio with decidable
emptiness. Although downclosures of lossy channel
languages are not effectively computable (an easy consequence of~\cite{mayr2003undecidable}), our algorithm employs
\cref{general:emptiness} to decide safety and termination.

\subsection{Boundedness}
\label{sub:gen_boundedness}
\begin{thm}\label{general:finiteness}
  Let $\cC$ be a full trio. 
The following are equivalent:
  \begin{enumerate}[(i)]
  \item\label{general:finiteness:boundedness} 
  Boundedness is decidable for asynchronous programs over $\cC$.
 \item\label{general:finiteness:finiteness} 
  Finiteness is decidable for $\cC$.
 \end{enumerate}
\end{thm}
\begin{proof}
Clearly, the construction for \impl{general:emptiness:safety}
{general:emptiness:emptiness} of \cref{general:emptiness} also works for
\impl{general:finiteness:boundedness} {general:finiteness:finiteness}: 
$\ap$ is unbounded iff $K$ is
infinite.

For the converse,
we first note that if finiteness is
decidable for $\cC$ then so is emptiness. Given $L\subseteq\Sigma^*$
from $\cC$, consider the
homomorphism $h\colon(\Sigma\cup\set{\lambda})^*\to\Sigma^*$ with
$h(\ltr{a})=\ltr{a}$ for every $\ltr{a}\in\Sigma$ and $h
(\lambda)=\varepsilon$. Then
$h^{-1}(L)$ belongs to $\cC$ and $h^{-1}(L)$ is finite if and only if $L$ is
empty: in the inverse homomorphism, $\lambda$ can be arbitrarily
inserted in any word. By \cref{general:emptiness}, this implies that
we can also decide safety. As a consequence of considering only full
trios, it is easy to see that the
problem of \emph{context reachability} reduces to safety: a
context $c=(d,\sigma,d')\in\fC$ is
\emph{reachable in $\ap$} if there is a reachable configuration
$(d,\mmap)$ in $\ap$ with $\mmap(\sigma)\ge 1$. This latter condition can be checked by moving from $d$ to a special state via a $\sigma$ transition. 

\begin{algorithm}[t]
  \SetKwBlock{Conc}{run concurrently}{end}
  \SetKwBlock{Begin}{begin}{end}
  \SetKw{Break}{break}
{\small
  \KwIn{Asynchronous program $\ap=(D,\Sigma,(L_c)_{c\in\fC},d_0,\mmap_0)$ over $\cC$}
  $F\leftarrow\emptyset$, $I\leftarrow\emptyset$\;
  \ForEach{context $c=(d,\sigma,d')\in\fC$}{
    \eIf{$L_c$ is infinite}{
      \If(\tcc*[f]{using algorithm for safety}){$c$ is reachable in $\ap$}{
        \Return{$\ap$ is unbounded.}
      }
      $I\leftarrow I\cup\{c\}$
    }{
      $F\leftarrow F\cup\{c\}$
    }
  }
  \ForEach{context $c\in F$}{
    \ForEach(\tcc*[f]{find a finite $A$ with $L_c=A$}){finite set $A\subseteq\Sigma^*$}{
      \If{$L_c\cap (\Sigma^*\setminus A)=\emptyset$ and $L_c\cap\{w\}\ne\emptyset$ for each $w\in A$}{
        $A_c\leftarrow A$\;
        \Break{}\tcc*{end inner \textbf{foreach}}
      }
    }
  }
  \ForEach{context $c\in I$}{
    $A_c\leftarrow\emptyset$
  }
  \eIf{$\ap'=(D,\Sigma,(A_c)_{c\in\fC},d_0,\mmap_0)$ is bounded}{
    \Return{$\ap$ is bounded.}
  }{
    \Return{$\ap$ is unbounded.}
  }
}
\caption{Checking Boundedness}\label{algboundedness}
\end{algorithm}

We now explain Algorithm \ref{algboundedness} for deciding boundedness of a given
aysnchronous program
$\ap=(D,\Sigma,(L_c)_{c\in\fC},d_0,\mmap_0)$. For every context $c$,
we first check if $L_c$ is infinite (feasible by assumption). This
paritions the set of contexts of $\ap$ into sets $I$ and $F$
which
are the contexts for which the corresponding language $L_c$ is
infinite and finite respectively. If any context in $I$ is
reachable, then $\ap$ is unbounded. Otherwise, all the reachable
contexts have a finite language. For every
finite language
$L_c$ for some $c \in F$, we explicitly find all the members of $L_c$.
This is
possible because any finite set $A$ can be checked with $L_c$ for
equality. $L_c\subseteq A$ can be checked by testing whether
$L_c\cap(\Sigma^*\setminus A)=\emptyset$ and
$L_c\cap(\Sigma^*\setminus A)$ effectively belongs to $\cC$. On the
other hand, checking $A\subseteq L_c$ just means checking whether
$L_c\cap\{w\}\ne\emptyset$ for each $w\in A$, which can be done the
same way. We can now construct asynchronous program $\ap'$ which
replaces all languages for contexts in $I$ by $\emptyset$ and replaces
those corresponding to $F$
by the explicit description. Clearly $\ap'$ is bounded iff $\ap$ is
bounded (since no contexts from $I$ are reachable) and the former can
be decided by \cref{thm:GantyM}. \qedhere
\end{proof}

We observe that boundedness is strictly harder than safety or
termination: There are full trios for which emptiness is decidable,
but finiteness is undecidable, such as the languages of reset vector addition
systems~\cite{DBLP:conf/icalp/DufourdFS98}
(see~\cite{thinniyam_et_al:LIPIcs:2019:11613} for a definition of the
language class) and
languages of lossy channel systems.

\subsection{Configuration reachability and liveness properties}
\label{sub:gen_undec}

\Cref{general:emptiness,general:finiteness} completely characterize
for which full trios safety, termination, and boundedness are
decidable. 
We turn to
configuration reachability, fair termination, and fair starvation. 
We suspect that it is unlikely that there is
a simple characterization of those language classes for which the
latter problems are decidable.
However, we show that they are decidable only for a limited range of
infinite-state systems.  
To this end, we prove that decidability of
any of these problems implies (a)~decidability of the others as well,
and also implies (b)~the decidability of a simple combinatorial
problem (called the $Z$-intersection problem) that is known to be undecidable for many expressive classes of languages.  

Let $Z\subseteq\{\ltr{a},\ltr{b}\}^*$ be the language
$Z=\{w\in\{\ltr{a},\ltr{b}\}^*\mid |w|_{\ltr{a}}=|w|_{\ltr{b}}\}$. The
\emph{$Z$-intersection problem} for a language class $\cC$ asks, given
a language $K\subseteq\{\ltr{a},\ltr{b}\}^*$ from $\cC$, whether
$K\cap Z\ne\emptyset$.
Informally, $Z$ is the language of all words with an equal number of $\ltr{a}$s
and $\ltr{b}$s and the $Z$-intersection problem asks if there is a word in $K$ with
an equal number of $\ltr{a}$s and $\ltr{b}$s.

\begin{thm}\label{general:intersection}
  Let $\cC$ be a full trio. The following statements are equivalent:
  \begin{enumerate}[(i)]
  \item\label{general:intersection:reach} Configuration reachability
  is decidable for asynchronous programs over $\cC$.
  \item\label{general:intersection:term} Fair termination is decidable
  for asynchronous programs over $\cC$.
  \item\label{general:intersection:starv} Fair starvation is decidable
  for asynchronous programs over $\cC$.
  \end{enumerate}
  Moreover, if decidability holds, then $Z$-intersection is decidable for $\cC$. 
\end{thm}
\begin{proof}
We prove \cref{general:intersection} by providing reductions among the
three problems and showing that $Z$-intersection reduces to
configuration reachability.  We use
diagrams similar to automata to describe asynchronous programs.
Here, circles represent global states of the program and we
draw an edge
\[
  \begin{tikzpicture}[gadget, node distance=1.5cm, baseline={([yshift=-.8ex]current bounding box.center)}]
    \node[state] (d) {$d$};
    \node[state] (d') [right=of d] {$d'$};
    \path (d) edge node {$\sigma|L$} (d');
  \end{tikzpicture}
\]
in case we have $L_{d,\sigma,d'}=L$ in our asynchronous program $\ap$.
Furthermore, we have $L_{d,\sigma,d'}=\emptyset$ whenever there is no
edge that specifies otherwise.  To simplify notation, we draw an edge
$\begin{tikzpicture}[gadget, node distance=1.5cm, baseline={([yshift=-.8ex]current bounding box.center)}]
    \node[state] (d) {$d$};
    \node[state] (d') [right=of d] {$d'$};
    \path (d) edge node {$w|L$} (d');
  \end{tikzpicture}$
in an asynchronous program for a word $w\in\Sigma^*$,
$w=\sigma_1\ldots\sigma_n$ with $\sigma_1,\ldots,\sigma_n\in\Sigma$,
to symbolize a sequence of states
  \begin{center}
    \begin{tikzpicture}[gadget, node distance=2.5cm]
      \node[state] (d) {$d$};
      \node[state] (2) [right=of d] {$2$};
      \node        (dots) [right=of 2] {$\cdots$};
      \node[state] (n) [right=of dots] {$n$};
      \node[state] (d') [right=of n] {$d'$};

      \path (d) edge node {$\sigma_1|\{\varepsilon\}$} (2);
      \path (2) edge node {$\sigma_2|\{\varepsilon\}$} (dots);
      \path (dots) edge node {$\sigma_{n-1}|\{\varepsilon\}$} (n);
      \path (n) edge node {$\sigma_{n}|L$} (d');
    \end{tikzpicture}
  \end{center}
  which removes
  $\multi{\sigma_1,\ldots,\sigma_n}$ from the task buffer and posts
  a multiset of handlers specified by the language $L$.

\noindent \textit{Proof of \impl{general:intersection:term}
{general:intersection:reach}} Given an asynchronous program
$\ap=(D,\Sigma,(L_c)_{c\in\fC},d_0,\mmap_0)$ and a configuration
$(d_f,\mmap_f)\in D\times\multiset{\Sigma}$, we construct
asynchronous program $\ap'$ as follows. Let $\ltr{z}$ be a fresh
letter and let $\mmap_f=\multi{\sigma_1,\ldots,\sigma_n}$. We obtain
$\ap'$ from $\ap$ by adding a new state $d'_f$ and including the
following edges:
\[
    \begin{tikzpicture}[gadget, node distance=3.5cm]
      \node[state] (df) {$d_f$};
      \node[state] (df') [right=of df] {$d'_f$};
      \path (df) edge node {$\ltr{z}\sigma_1\cdots\sigma_n|\{\ltr{z}\}$} (df');
      \path (df') edge [loop right] node {$\ltr{z}|\{\ltr{z}\}$} (df');
    \end{tikzpicture}
\]
Starting from $(d_0,\mmap_0\oplus\multi{\ltr{z}})$, the program
$\ap'$ has a fair infinite run iff $(d_f,\mmap_f)$ is reachable in $\ap$.
The `if' direction is obvious.
Conversely, $\ltr{z}$ has to be executed in any fair run $\rho$ of
$\ap'$
which
implies that $d'_f$ is reached by $\ap'$ in $\rho$. Since only $
\ltr{z}$ can be executed at $d'_f$ in $\rho$, this means that the
multiset is exactly $\mmap_f$ when $d_f$ is reached during $\rho$.
Clearly this initial segment of $\rho$ corresponds to a run of $\ap$
which reaches the target configuration.\\
\noindent\textit{Proof of \impl{general:intersection:starv}
{general:intersection:term}} We construct
$\ap'=(D,\Sigma',(L'_c)_{c\in\fC'},d_0,\mmap'_0)$ given \[\ap=(D,\Sigma,
(L_c)_{c\in\fC},d_0,\mmap_0)\] over $\cC$
as follows. Let $\Sigma'=\Sigma\cup\{\ltr{s}\}$, where $\ltr{s}$ is
a fresh handler. Replace each edge 
\[
  \begin{tikzpicture}[gadget, node distance=1.5cm, baseline={([yshift=-.8ex]current bounding box.center)}]
    \node[state] (d) {$d$};
    \node[state] (d') [right=of d] {$d'$};
    \path (d) edge node {$\sigma|L$} (d');
  \end{tikzpicture}
\quad \text{by} \quad
\begin{tikzpicture}[gadget, node distance=2.5cm, baseline={([yshift=-.8ex]current bounding box.center)}]
    \node[state] (d) {$d$};
    \node[state] (d') [right=of d] {$d'$};
        \path (d) edge node {$\sigma|L \cup L \ltr{s}$} (d');

    \path (d') edge [loop right] node {$\ltr{s}|\varepsilon$} (d');
\end{tikzpicture}%
\]
at every state $d\in D$. Moreover, we set
$\mmap'_0=\mmap_0\oplus\multi{\ltr{s}}$. Then $\ap'$
has an infinite fair run that starves some handler if and only if
$\ap$ has an infinite fair run.  From an infinite fair run $\rho$ of
$\ap$, we obtain an infinite fair run of $\ap'$ which starves $
\ltr{s}$, by producing $\ltr{s}$ while simulating $\rho$ and consuming it in the
loop. Conversely, from an infinite fair run $\rho'$ of $\ap'$ which
starves some $\tau$, we obtain an infinite fair run $\rho$ of $\ap$ by
omitting all productions and consumptions of $\ltr{s}$ and removing
two extra instances of $\ltr{s}$ from all configurations.\\
\noindent\textit{Proof of \impl{general:intersection:reach}
{general:intersection:starv}} From
$\ap=(D,\Sigma,(L_c)_{c\in\fC},d_0,\mmap_0)$ over $\cC$, for each
subset $\Gamma\subseteq\Sigma$ and $\tau\in\Sigma$, we construct an
asynchronous program
$\ap_{\Gamma,\tau}=(D',\Sigma',(L_c)_{c\in\fC'},d'_0,\mmap'_0)$ over
$\cC$ such that a particular configuration is reachable in
$\ap_{\Gamma,\tau}$ if and only if $\ap$ has a fair infinite run
$\rho_
{\Gamma,\tau}$, where $\Gamma$ is the set of handlers that is executed
infinitely often
in $\rho_{\Gamma,\tau}$ and $\rho_{\Gamma,\tau}$ starves $\tau$. Since
there are only finitely many choices for $\Gamma$ and $\tau$,
decidability of configuration reachability implies decidability of
fair starvation. The idea is that a run $\rho_{\Gamma,\tau}$ exists if
and only if
there exists a run
\begin{equation}
  (d_0,\mmap_0)\autstep[\sigma_1]\cdots\autstep[\sigma_n](d_n,\mmap_n)=(e_0,\nmap_0)\autstep[\gamma_1](e_1,\nmap_1)\autstep[\gamma_2]\cdots\autstep[\gamma_k](e_k,\nmap_k), \label{run:twophases}\end{equation}
where $\bigcup_{i=1}^k\{\gamma_i\}=\Gamma$, for each $1 \leq i \leq k$
$\nmap_i \in \multiset{\Gamma}$,
$\mmap_n\preceq\nmap_k$, and for each $i\in\{1,\ldots,k\}$
with $\gamma_i=\tau$, we have $\nmap_{i-1}(\tau)\ge 2$. In such a run, we call $(d_0,\mmap_0)\autstep[\sigma_1]\cdots\autstep[\sigma_n](d_n,\mmap_n)$ its \emph{first phase} and $(e_0,\nmap_0)\autstep[\gamma_1]\cdots\autstep[\gamma_k](e_k,\nmap_k)$ its \emph{second phase}. 

Let us explain how $\ap_{\Gamma,\tau}$ reflects the existence of a run
as in \cref{run:twophases}. The set $\Sigma'$ of handlers of
$\ap_{\Gamma,\tau}$ includes $\Sigma$, $\bar{\Sigma}$ and
$\hat{\Sigma}$, where
$\bar{\Sigma}=\{\bar{\sigma}\mid\sigma\in\Sigma\}$ and
$\hat{\Sigma}=\{\hat{\sigma}\mid\sigma\in\Sigma\}$ are disjoint copies
of $\Sigma$. This means, a multiset
$\multiset{\Sigma'}$ contains multisets
$\mmap'=\mmap\oplus\bar{\mmap}\oplus\hat{\mmap}$ with
$\mmap\in\multiset{\Sigma}$, $\bar{\mmap}\in\multiset{\bar{\Sigma}}$,
and $\hat{\mmap}\in\multiset{\hat{\Sigma}}$. A run of
$\ap_{\Gamma,\tau}$ simulates the two phases of $\rho$. While
simulating the first phase, $\ap_{\Gamma,\tau}$ keeps two copies of
the task buffer, $\mmap$ and $\bar{\mmap}$. The copying is easily
accomplished by a homomorphism with $\sigma\mapsto\sigma
\bar{\sigma}$ for each $\sigma \in \Sigma$. At some point,
$\ap_{\Gamma,\tau}$ switches into simulating the second phase. There,
$\bar{\mmap}$ remains unchanged, so that it stores the value of
$\mmap_n$ in \cref{run:twophases} and can be used in the end to make
sure that $\mmap_n\preceq \nmap_k$.

Hence, in the second phase, $\ap_{\Gamma,\tau}$ works, like $\ap$,
only with $\Sigma$. However, whenever a handler $\sigma\in\Sigma$ is
executed, it also produces a task $\hat{\sigma}$. These handlers are
used at the end to make sure that every $\gamma\in\Gamma$ has been
executed at least once in the second phase. Also, whenever
$\tau$ is executed, $\ap_{\Gamma,\tau}$ checks that at least two
instances of $\tau$ are present in the task buffer, thereby ensuring
that $\tau$ is starved.

In the end, a distinguished final state allows $\ap_{\Gamma,\tau}$ to
execute handlers in $\Gamma$ and $\bar{\Gamma}$ simultaneously to make
sure that $\mmap_n\preceq\nmap_k$. In its final state,
$\ap_{\Gamma,\tau}$ can execute handlers $\hat{\gamma}\in\hat{\Gamma}$
and $\gamma\in\Gamma$ (without creating new handlers).  In the final
configuration, there can be no $\hat{\sigma}$ with
$\sigma\in\Sigma\setminus\Gamma$, and there has to be exactly one
$\hat{\gamma}$ for each $\gamma\in\Gamma$. This guarantees that
(i)~each handler in $\Gamma$ is executed at least once during the
second phase, (ii)~every handler executed in the second phase is from
$\Gamma$, and (iii)~$\mmap_n$ contains only handlers from $\Gamma$
(because handlers from $\bar{\Sigma}$ cannot be executed in the second
phase).\\

Let us now describe $\ap_{\Gamma,\tau}$ in detail. We have 
$\Sigma'=\Sigma\cup\bar{\Sigma}\cup\hat{\Gamma}\cup\{\ltr{z}\}$, where 
$\ltr{z}$ is a fresh letter.  The set of states is 
$D'=D\cup\tilde{D}\cup\{d_f\}$, where 
$\tilde{D}=\{\tilde{d}\mid d\in D\}$ is a disjoint copy of $D$ for 
simulating the second phase, and $d_f$ is a fresh state. Moreover, we 
change the languages $L_c$ in the following way: 
\[
  \begin{tikzpicture}[gadget, node distance=1.5cm, baseline={([yshift=-.8ex]current bounding box.center)}]
    \node[state] (d) {$d$}; 
    \node[state] (d') [right=of df] {$d'$}; 
    \path (d) edge node {$\sigma|L$} (d'); 
  \end{tikzpicture}
    ~~\leadsto~~
  \begin{tikzpicture}[gadget, node distance=2.5cm, baseline={([yshift=-.8ex]current bounding box.center)}]
    \node[state] (d) {$d$}; 
    \node[state] (d') [right=of df] {$d'$}; 
    \path (d) edge node {$\sigma\bar{\sigma}|h^{\pm}(L)$} (d'); 
  \end{tikzpicture}
\]
where $h^{\pm}\colon\Sigma^*\to(\Sigma\cup\bar{\Sigma})^*$ is the 
homomorphism with $\sigma\mapsto\sigma\bar{\sigma}$ for every 
$\sigma\in\Sigma$.  This means, for every context $c=(d,\sigma,d')$, 
we have $d\autstep[\sigma\bar{\sigma}|h^{\pm}(L)]d'$ in 
$\ap_{\Gamma,\tau}$.  Note that since $\cC$ is a full trio, it is in 
particular closed under homomorphisms. Hence, $h^\pm(L)$ belongs to 
$\cC$.  Moreover, $\ap_{\Gamma,\tau}$ can spontaneously switch to 
simulating the second phase: For each $d\in D$, we have 
\[
  \begin{tikzpicture}[gadget, node distance=2cm]
    \node[state] (d) {$d$}; 
    \node[state] (dtilde) [right=of df] {$\tilde{d}$}; 
    \path (d) edge node {$\ltr{z}|\{\ltr{z}\}$} (dtilde); 
  \end{tikzpicture}
\]
Here, the handler $\ltr{z}$ merely allows us to move to state $\tilde{d}$
without interfering with the other handlers. 
In the second phase, $\ap_{\Gamma,\tau}$ simulates $\ap$ slightly differently. We perform the following replacement: 
\begin{equation}
  \begin{tikzpicture}[gadget, node distance=1.5cm, baseline={([yshift=-.8ex]current bounding box.center)}]
    \node[state] (d) {$d$}; 
    \node[state] (d') [right=of df] {$d'$}; 
    \path (d) edge node {$\sigma|L$} (d'); 
  \end{tikzpicture}
  ~~\leadsto~~
  \begin{cases}
  \begin{tikzpicture}[gadget, node distance=2.5cm, baseline={([yshift=-.8ex]current bounding box.center)}]
    \node[state] (d) {$\tilde{d}$}; 
    \node[state] (d') [right=of df] {$\tilde{d}'$}; 
    \path (d) edge node {$\sigma|L\hat{\sigma}$} (d'); 
  \end{tikzpicture} & \text{if $\sigma\ne\tau$} \\
  \begin{tikzpicture}[gadget, node distance=2.5cm, baseline={([yshift=-.8ex]current bounding box.center)}]
    \node[state] (d) {$\tilde{d}$}; 
    \node[state] (d') [right=of df] {$\tilde{d}'$}; 
    \path (d) edge node {$\tau\tau|L\tau\hat{\tau}$} (d'); 
  \end{tikzpicture} & \text{if $\sigma=\tau$}
\end{cases}
\label{general:twoedges}
\end{equation}
Note that since $\cC$ is a full trio, the languages $Lv=\{uv\mid u\in L\}$
belong to $\cC$ for every word $v$. 
Finally, $\ap_{\Gamma,\tau}$ can spontaneously switch to its 
distinguished state $d_f$, so that we have for every 
$\tilde{d}\in \tilde{D}$ and every $\gamma\in\Gamma$: 
\begin{equation}
  \begin{tikzpicture}[gadget, node distance=2cm]
    \node[state] (dtilde) {$\tilde{d}$}; 
    \node[state] (df) [right=of df] {$d_f$}; 
    \path (dtilde) edge node {$\ltr{z}|\{\varepsilon\}$} (df); 
    \path (df) edge [loop above] node  {$\gamma\bar{\gamma}|\{\varepsilon\}$} (df); 
    \path (df) edge [loop right] node  {$\gamma|\{\varepsilon\}$} (df); 
    \path (df) edge [loop below] node {$\hat{\gamma}|\{\varepsilon\}$} (df); 
  \end{tikzpicture}
  \label{general:intersection:finalloops}
\end{equation}
The initial configuration of $\ap_{\Gamma,\tau}$ is $(d_0,\mmap'_0)$, 
where $\mmap'_0=\mmap_0\oplus\bar{\mmap}_0\oplus\multi{\ltr{z}}$, 
where $\bar{\mmap}_0$ is obtained from $\mmap_0$ by replacing every 
occurrence of $\sigma\in\Sigma$ in $\mmap_0$ with 
$\bar{\sigma}$.  The final configuration is $(d_f,\mmap_f)$, where 
$\mmap_f\in\multiset{\hat{\Gamma}}$ is the multiset with 
$\mmap_f(\hat{\gamma})=1$ for every $\gamma\in\Gamma$. It is now 
straightforward to check that $(d_f,\mmap_f)$ is reachable in 
$\ap_{\Gamma,\tau}$ if and only if $\ap$ has an infinite fair run that 
starves $\tau$. 

\noindent\textit{Decidability of $Z$-intersection} To complete the
proof of
\cref{general:intersection}, we reduce
$Z$-intersection to configuration reachability. Given
$K\subseteq\{\ltr{a},\ltr{b}\}^*$ from $\cC$, we construct the
asynchronous program $\ap=(D,\Sigma,(L_c)_{c\in\fC},d_0,\mmap_0)$ over
$\cC$ where $D=\{d_0,0,1\}$,
$\Sigma=\{\ltr{a},\ltr{b}, \ltr{c}\}$, by including the following
edges:
\[
  \begin{tikzpicture}[gadget, node distance=2cm]
  \node[state] (d0) [initial] {$d_0$};
  \node[state] (0) [right=of d0] {$0$};
  \node[state] (1) [right=of 0] {$1$};
  \path (d0) edge node {$\ltr{c}|K$} (0);
  \path (0) edge [bend left] node {$\ltr{a}|\{\varepsilon\}$} (1);
  \path (1) edge [bend left] node {$\ltr{b}|\{\varepsilon\}$} (0);
\end{tikzpicture}
\]
The initial task buffer is $\mmap_0=\multi{\ltr{c}}$.
Then clearly, the configuration $(0,\multi{})$ is reachable in
$\ap$ if and only if $K\cap Z\ne\emptyset$.

If the construction seems abstract, recall the example from Section~\ref{sec:preliminaries}:
the procedure \texttt{s1} plays the role of $K$ and generates strings from its 
language in $\set{\ltr{a}, \ltr{b}}^*$.
The procedures $\ltr{a}$ and $\ltr{b}$ take turns to ensure there is an equal number of them;
the states $0$ and $1$ are values of \texttt{turn}.~\qedhere
\end{proof}

Theorem~\ref{general:intersection} is useful in the contrapositive to
show undecidability. For example, one can 
show undecidability of $Z$-intersection for languages of lossy
channel systems (see~\cref{sub:gen_safety_and_termination}): One
expresses reachability in a non-lossy FIFO system by making sure that
the numbers of enqueue- and dequeue-operations match.
Thus, for asynchronous programs over lossy channel systems, the problems of
\cref{general:intersection} are undecidable.
We also use \cref{general:intersection} in \cref{sec:higher-order} to conclude
undecidability for higher-order asynchronous programs, already at order $2$.

\section{Higher-Order Asynchronous Programs}\label{sec:higher-order}

We apply our general decidability results to
asynchronous programs over (deterministic) higher-order recursion
schemes ($\hors$).
Kobayashi~\cite{DBLP:conf/popl/Kobayashi09} has shown how
higher-order functional programs can be modeled using $\hors$.
In his setting, a program contains instructions
that access certain resources.  The path language of the $\hors$ 
produced by Kobayashi is the set of possible sequences of
instructions.  
For us, the input program contains \texttt{post} instructions (which post handlers i.e. add tasks to the task buffer) 
and we translate higher-order programs with \texttt{post} 
instructions into a $\hors$ whose path language is used as the language of handlers.

We recall some definitions from~\cite{DBLP:conf/popl/Kobayashi09}. 
The set of \textit{types} is defined by the grammar
$A\ :=\ \otype \mid A \rightarrow A$.
The \textit{order} $\ord(A)$ of a type $A$ is inductively defined as
$\ord(\otype)=0$ and $\ord(A \rightarrow B) := \max(\ord(A)+1,\ord(B))$.
The \textit{arity} of a type is inductively defined by $\arity(\otype)=0$
and $\arity(A \rightarrow B)=\arity(B)+1$.
We assume a countably infinite set $\Var$ of typed variables $x: A$.
For a set $\Theta$ of typed symbols, the set $\appterms$ of
\textit{terms} generated from $\Theta$ is the least set which
contains $\Theta$ such that whenever $s:A\to B$ and $t:A$ belong to
$\appterms$, then also $s \, t:B$ belongs to $\appterms$.
By convention the type $\otype \rightarrow \ldots (\otype \rightarrow (\otype \rightarrow
\otype))$ is written $\otype \rightarrow \ldots\rightarrow \otype \rightarrow \otype$ and
the term $((t_1 t_2)t_3\cdots)t_n$ is written $t_1 t_2\cdots t_n$.
We write $\bar{x}$ for a sequence $(x_1,x_2,\ldots, x_n)$ of variables.

A higher-order recursion scheme ($\hors$) is a tuple
 $\rs=
 (\Sigma,\nonterm,\rewrite,S)$ where 
  $\Sigma$ is a set of typed \textit{terminal} symbols of types
  of order 0 or 1,
  $\nonterm$ is a set of typed \textit{non-terminal} symbols (disjoint
  from terminal symbols), 
  $S:\otype$ is the start non-terminal symbol and
 $\rewrite$ is a set of rewrite rules $F x_1 x_2 \cdots x_n
 \twoheadrightarrow t$ where $F:A_1
 \rightarrow \cdots \rightarrow A_n \rightarrow \otype$ is a non-terminal in $\nonterm$,
 $x_i:A_i$ for all $i$ are variables and $t:\otype$ is a term
 generated from $\Sigma \cup \nonterm \cup \{ x_1,x_2,\ldots, x_n\}$.
 The order of a $\hors$ is the maximum order of a non-terminal symbol.
We define a rewrite relation 
$\twoheadrightarrow$ on terms over $\Sigma \cup \nonterm$ as follows:
$F \bar{a} \twoheadrightarrow t[\bar{x} / \bar{a}]$ if $F \bar{x}
 \twoheadrightarrow t \in \rewrite$, %
and 
  if $t \twoheadrightarrow t'$ then $ts \twoheadrightarrow t's$ and 
  $st \twoheadrightarrow st'$.
The reflexive, transitive closure of 
$\twoheadrightarrow$ is denoted 
$\twoheadrightarrow^*$. 
A \textit{sentential form} $t$ of $\rs$
is a term over $\Sigma \cup \nonterm$ such that 
$S \twoheadrightarrow^* t$.

If $N$ is the maximum arity of a symbol in $\Sigma$, then a (possibly infinite) tree over $\Sigma$ 
is a partial function $tr$ from $\{ 0, 1,\ldots, N-1\}^*$ to $\Sigma$ that fulfills the following conditions:
  $\varepsilon \in \dom(tr)$,
  $\dom(tr)$ is closed under prefixes, and
  if $tr(w)=a$ and $\arity(a)=k$ then $\{ j \mid wj \in \dom(tr)\}= \{ 0,1,\ldots,k-1\}$. 

  A \emph{deterministic} $\hors$ is one where there is exactly one rule
  of the form
  $F
  x_1 x_2 \cdots x_n \rightarrow t$
  for every non-terminal $F$.
Following~\cite{DBLP:conf/popl/Kobayashi09}, we show how a deterministic
 $\hors$ can be used to represent a higher-order
 pushdown language arising from a higher-order functional program. 
Sentential forms can be seen as trees over $\Sigma \cup
  \nonterm$. 
A sequence $\Pi$ over $\{0,1,\ldots,N-1\}$ is a \textit{path} of $tr$ if
every finite prefix of $\Pi \in \dom(tr)$ and $\Pi$ is maximal. The set of finite paths in a tree
$tr$ will be denoted $\paths(tr)$. Unless otherwise specified, a path is finite in our context. Associated with any path
$\Pi=n_1,n_2,\ldots,n_k$ is the word 
\[ w_{\Pi}=tr(n_1)tr(n_1n_2)\cdots tr(n_1n_2\cdots n_k). \]
We will also write $w_{\Pi}(\bar{m})$ to denote the letter $tr(\bar{m})$ where $\bar{m}=n_1n_2\ldots n_l$ for some $l \leq k$.
We associate a  \textit{value tree} $\tree_{\rs}$ associated with a deterministic
$\hors$ $\rs$ in the following way.
For a sentential form $t$ define the finite tree $t^{\bot}$ by case
analysis:
$t^{\bot} = f$ if $f$ is a terminal symbol and
$t^{\bot} = t_1^{\bot} t_2^{\bot}$ if $t = t_1 t_2$ and $t_1^{\bot} \neq \bot$,
and $t^{\bot} = \bot$ otherwise.

Intuitively, the tree $t^{\bot}$ is obtained by replacing any subterm
whose root is a non-terminal by the nullary symbol $\bot$. We define
the partial order
$\leq_t$ on $\dom(\Sigma) \cup \{ \bot\}$ by $\bot \leq_t a$ for all
$a \in \dom(\Sigma)$, which is extended to trees by 
\[ t \leq_t s
\iff \forall \bar{n} \in \dom(t): \bar{n} \in \dom(s) \land t(
\bar{n}) \leq_t s(\bar{n})
\]

The value tree generated by a $\hors$ $\rs$ is denoted
$\tree_{\rs}$ and is obtained by repeated application of the rewrite
rules to the start symbol $S$. 
To make this formal, we write $\lub(T)$ for the
least upper
bound of a collection $T$ of trees with respect to the order $\leq_t$.
Then we set 
$\tree_{\rs} := \lub(\{ t^{\bot} \mid S \rightarrow^* t\})$.
  Any $\hors$ for which the value tree $\tree_{\rs}$ does not contain any $\bot$ symbol is
  called \emph{productive}. All $\hors$ dealt with in this paper are
  assumed to be productive. Checking if 
  a given $\hors$ is productive is decidable, see, e.g.,~\cite{grellois2016semantics}.

 Let $\Sigma_1:= \{a \in \Sigma \mid \arity(a)=1 \}$. The 
 \textit{path
 language} $\langpath(\rs)$ of a deterministic $\hors$ $\rs$ is
 defined as $\{ \proj_{\Sigma_1}(w_{\Pi}) \mid \Pi \in \paths(\tree_
 {\rs}) \}$.
 The \textit{tree language} $\langtree(\rs)$ associated with a (nondeterministic) $\hors$
 is the set of sentential forms of $\rs$ which contain only labels from $\Sigma$.

The deterministic $\hors$ corresponding to the higher-order function $
\mathtt{s4}$ from Figure~\ref{fig:ex}
is given by $\rs= (\Sigma,\nonterm,\rewrite,S)$, where 
\begin{flalign*}
  \Sigma =&\{ \ltr{br}:\otype \rightarrow \otype \rightarrow 
  \otype, \ltr{c},\ltr{d},\ltr{f}:\otype \rightarrow \otype,
  \ltr{e}:\otype\}\\
  \nonterm =& \{S:\otype, F:(\otype\rightarrow \otype)
  \rightarrow 
  \otype \rightarrow \otype, H:(\otype\rightarrow \otype) \rightarrow 
  \otype \rightarrow \otype, I:\otype \rightarrow \otype\} \\
  \rewrite=& \{ S \twoheadrightarrow F \; I \; \ltr{e}, I \; x
  \twoheadrightarrow
  x, F \; G \; x \twoheadrightarrow \ltr{br} (F \; (H \; G) \; (
  \ltr{f}x)
  ) \; (G\; x), H \; G \; x \twoheadrightarrow \ltr{c}(G(\ltr{d}x)) 
  \}
\end{flalign*}
The path language $\langpath(\rs) = \set{\mathtt{c}^n \mathtt{d}^n \mathtt{f}^n \mid n\geq 0}$.
To see this, apply the reduction rules to get the value tree $\tree_{\rs}$ shown on the right: 

\begin{minipage}{0.60\textwidth}
{\small
  \begin{flalign*}
    S &\twoheadrightarrow F \; I \; \ltr{e} \twoheadrightarrow \ltr{br} \; (F \; (HI) \; (\ltr{fe})) \; (I\ltr{e})\\
    & \twoheadrightarrow \ltr{br} \; (F \; (HI) \; (\ltr{fe})) \; \ltr{e}\\
    & \twoheadrightarrow \ltr{br} \; (\ltr{br} \; (F \; (H(HI)) \; (\ltr{f}(\ltr{fe}))) \; 
    (HI(\ltr{fe}))) \; \ltr{e}\\
    & \twoheadrightarrow \ltr{br} \; (\ltr{br} \; (F \; (H(HI)) \; (\ltr{f}(\ltr{fe}))) \;(\ltr{c}(I(\ltr{d(fe)})))) \;
    \ltr{e}\\
    & \twoheadrightarrow \ltr{br} \; (\ltr{br} \; (F \; (H(HI)) \; (\ltr{f}(\ltr{f e}))) \;
    (\ltr{c(d(fe))})) \;
    \ltr{e}\\
    &\twoheadrightarrow \cdots
  \end{flalign*}
}
\end{minipage}
\quad %
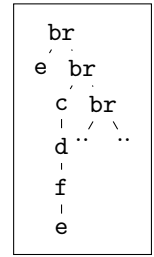
\begin{minipage}{0.25\textwidth}

\begin{centering}
{\small
\begin{tikzpicture}[framed,scale=0.85]
\label{tikz:value_tree_cndnfn}
\def \d {0.65cm}%
  \node(v1) at (0,0) {$\ltr{br}$};

  \node(v2) at ({-120}:\d) {$\ltr{e}$};
  \node(v3) at ({-60}:\d) {$\ltr{br}$};

  \node(v4) at ($({-120}:\d)+(v3)$) {$\ltr{c}$};
  \node(v5) at ($({-60}:\d)+(v3)$) {$\ltr{br}$};

  \node(v6) at ($(v4)-(0,\d)$) {$\ltr{d}$};
  \node(v7) at ($(v6)-(0,\d)$) {$\ltr{f}$};
\node(v8) at ($(v7)-(0,\d)$) {$\ltr{e}$};

\node(v9) at ($(v5)+({-120}:\d)$) {$..$};
\node(v10) at ($(v5)+({-60}:\d)$) {$..$};

\draw (v1) -- (v2);
\draw (v1) -- (v3);
\draw (v3) -- (v4);
\draw (v3) -- (v5);
\draw (v4) -- (v6);
\draw (v6) -- (v7);
\draw (v7) -- (v8);
\draw (v5) -- (v9);
\draw (v5) -- (v10);
\end{tikzpicture}
}
\end{centering}

\end{minipage}
  A $\hors$ $\rs$ is called a \emph{word scheme} if it has exactly
  one nullary terminal symbol $\e$ and all other terminal symbols $\tilde{\Sigma}$
  are of arity
  one.  The \textit{word language} $\langword(\rs) \subseteq \tilde{\Sigma}^*$
  defined by $\rs$ is $\langword(\rs)=\{ a_1a_2 \cdots a_n \mid (a_1(a_2\cdots(a_n
  (\e))\cdots)) \in \langtree(\rs)\}$. 
  We denote by $\langhors$ the class of languages
  $\langword(\rs)$ that occur as the word language of a higher-order
  recursion scheme $\rs$.
Note that path languages and languages of word schemes are both word languages
over the set $\tilde{\Sigma}$ of unary symbols considered as letters.
They are connected by the following proposition, a proof of which is
given in Appendix \ref{sec:appendix_ho}.\footnote{
  The models of $\hors$ (used in model checking higher-order programs~\cite{DBLP:conf/popl/Kobayashi09})
  and word schemes (used in language-theoretic exploration of downclosures~\cite{HagueKO16,ClementePSW16})
  are somewhat different. 
  Thus, we show an explicit reduction between the two formalisms.
}

\begin{prop}
\label{prop:pathslangs_to_word_schemes}
  For every order-$n$ $\hors$ $\rs=(\Sigma,\nonterm, S, \rewrite)$
  there exists
  an order-$n$ word scheme $\rs'=(\Sigma',\nonterm', S', \rewrite')$ such
  that
  $\langpath(\rs)=\langword(\rs')$.
\end{prop}

A consequence of~\cite{DBLP:conf/popl/Kobayashi09} 
and Prop.~\ref{prop:pathslangs_to_word_schemes} 
is that the ``post'' language of higher-order functional programs (i.e. the language corresponding to sequences of newly spawned tasks)
can be modeled as the language of a word scheme.
Hence, we define an \emph{asynchronous program over $\hors$} as an
asynchronous program over the language class $\langhors$ and we can use the following results
on word schemes.

\begin{thm}\label{thm:hors}
$\hors$ and word schemes form effective full trios~\cite{ClementePSW16}.
  Emptiness~\cite{KobayashiO2011} and finiteness~\cite{Parys2017} of order-$n$ word schemes are $\EXPTIME[(n-1)]$-complete.
\end{thm}

Now \cref{general:emptiness,general:finiteness}, together with
\cref{prop:pathslangs_to_word_schemes} imply the decidability results in \cref{coro:hodec}.
The undecidability result is a consequence of \cref{general:intersection} and the 
undecidability of the $Z$-intersection problem for indexed languages or equivalently,
order-2 pushdown automata as shown in~\cite{Zetzsche2015b}.
Order-2 pushdown automata can be effectively turned into order-2
OI~grammars~\cite{DammGoerdt1986}, which in turn can be translated
into order-2 word schemes~\cite{damm1982io}. 
A direct (and independent) proof of undecidability of the $Z$-intersection problem for order-2 word schemes has appeared in~\cite[Theorem 4]{kobayashi2019inclusion}.

\begin{cor}
\label[cor]{coro:hodec}
For asynchronous programs over $\hors$: 
(1)~Safety, termination, and boundedness are decidable. 
(2)~Configuration reachability, fair termination, and fair starvation
  are undecidable already at order-2. 
\end{cor}

\section{A Direct Algorithm and Complexity Analysis}

We say that \emph{downclosures are computable} for a
language class $\cC$ if for a given description of a language $L$ in
$\cC$, one can compute an automaton for the regular language
$\dcl{L}$. 
A consequence of Proposition~\ref{dcpreservation} and
Theorem~\ref{thm:GantyM} is that if one can compute downclosures for
some language class, then one can avoid the enumerative approaches of
Section~\ref{general} and get a ``direct algorithm.''  
The direct algorithm replaces each handler by its
downclosure and then invokes the decision procedure summarized in
Theorem~\ref{thm:GantyM}.

\subsection{Higher-Order Recursion Schemes}

The direct algorithm for asynchronous programs over $\hors$ 
relies on the recent breakthrough results on computing
downclosures. 
\begin{thmC}[\cite{Zetzsche2015b,HagueKO16,ClementePSW16}]
  Downclosures are effectively computable for $\langhors$.
\end{thmC}

Unfortunately, current techniques do not yet provide a complexity upper bound 
based on the above theorem.
To understand why, we describe the current algorithm for computing downclosures.
In~\cite{Zetzsche2015b}, it was
shown that in a full trio $\cC$, downclosures are computable if and
only if the \emph{diagonal problem} for $\cC$ is decidable. 
The latter asks, given a language $L\subseteq\Sigma^*$, whether for every
$k\in\nats$, there is a word $w\in L$ with $|w|_\sigma\ge k$ for every
$\sigma\in\Sigma$. 
The diagonal problem was then shown to be decidable
for higher-order pushdown automata~\cite{HagueKO16} and then for word schemes~\cite{ClementePSW16}. 
However, the algorithm from~\cite{Zetzsche2015b}
to compute downclosures using an oracle for the diagonal problem 
employs enumeration to compute a downclosure automaton.
Thus, the enumeration is hidden inside the downclosure computation.

We conjecture that downclosures can in fact be
computed in elementary time (for word schemes of fixed order). 
This would imply an elementary time procedure for asynchronous programs over $\hors$ as well.

\subsection{Context Free Languages}

For handlers over context-free languages, given e.g., as pushdown automata,
Ganty and Majumdar show an $\EXPSPACE$ upper bound.
Precisely, the algorithm of~\cite{GantyM12} constructs for each handler
a polynomial-size Petri net with certain guarantees (forming a 
so-called \emph{adequate family of Petri nets}) that accepts a
Parikh equivalent language.  
These Petri nets are then used to construct a larger Petri net,
polynomial in the size of the asynchronous program and the adequate family
of Petri nets, in which the respective property (safety,
boundedness, or termination) can be phrased as a query decidable in
\EXPSPACE. 
A natural question is whether there is a downclosure-based algorithm with the same
asymptotic complexity.

Our goal is to replace the Parikh-equivalent Petri nets with Petri nets
recognizing the downclosure of a language.
It is an easy consequence of \cref{dcpreservation} that
the resulting Petri nets can be used in place of the adequate families
of Petri nets in the procedures for safety, termination, and
boundedness of~\cite{GantyM12}. 
Thus, if we can ensure these Petri nets are polynomial in the size of the handler,
we get an $\EXPSPACE$ upper bound.
Unfortunately, a finite automaton for $\dcl{L}$ may require exponentially many states in the
pushdown automaton~\cite{BachmeierLS2015}.
Thus a naive approach gives a $\TWOEXPSPACE$ upper bound.

We show here that for each context-free language $L$,
one can construct in polynomial time a $1$-bounded Petri net accepting
$\dcl{L}$. 
(Recall that a Petri net is $1$-bounded if every reachable marking has
at most one token in each place.)
This is a language theoretic result of independent interest.
When used in the construction of~\cite{GantyM12}, this matches the $\EXPSPACE$ upper bound.

As a byproduct, our translation yields a simple direct construction of
a finite automaton for $\dcl{L}$ when $L$ is given as a pushdown
automaton. This is of independent interest because earlier
constructions of $\dcl{L}$ always start from a context-free
grammar and produce (of necessity!) exponentially large 
NFAs~\cite{vanLeeuwen1978,Courcelle1991,BachmeierLS2015}.

We begin with some preliminary definitions.

\paragraph{Pushdown automata} If $\Gamma$ is an alphabet, we write
$\bar{\Gamma}=\{\bar{\gamma} \mid \gamma\in\Gamma\}$.  Moreover, if
$x=\bar{\gamma}$, then we define $\bar{x}=\gamma$. For a word
$v\in(\Gamma\cup\bar{\Gamma})^*$, $v=v_1\cdots v_n$,
$v_1,\ldots,v_n\in\Gamma\cup\bar{\Gamma}$, we set
$\bar{v}=\bar{v_n}\cdots\bar{v_1}$.  A \emph{pushdown automaton} is a tuple
$\cA=(Q,\Sigma,\Gamma,E,q_0,q_f)$, where $Q$ is a finite set of \emph{states},
$\Sigma$ is its \emph{input alphabet}, $\Gamma$ is its \emph{stack alphabet},
$E$ is a finite set of \emph{edges}, $q_0\in Q$ is its \emph{initial state},
and $F\subseteq Q$ is its set of \emph{final states}. An edge is a four-tuple
$(p,R,v,q)$, where $p,q\in Q$, $R\subseteq\Sigma^*$ is a regular language, and
$v\in \Gamma\cup\bar{\Gamma}\cup\{\varepsilon\}$.  We also write
\[ p\xrightarrow{R|v}q \]
to denote an edge $(p,R,v,q)\in E$. Intuitively, it tells us that from state
$q$, we can read any word in $R$ as input and modify the stack as specified by
$v$. Here, $v\in\Gamma$ means we push $v$ onto the stack. Moreover, $v\in\Gamma$,
$v=\bar{\gamma}$, means we pop $\gamma$ from the stack. Finally, $v=\varepsilon$
means we do not change the stack.
Let us make this formal. A \emph{configuration} of $\cA$ is a pair $(q,w)$ with
$q\in Q$ and $w\in\Gamma^*$.  For configurations $(q,w)$ and $(q',w')$, we
write $(q,w)\autstep[u](q',w')$ if there is an edge $(q,R,v,q')$ in $\cA$ such
that $u\in R$ and (i)~if $v=\varepsilon$, then $w'=w$, (ii)~if $v\in\Gamma$,
then $w'=wv$ and (iii)~if $v=\bar{\gamma}$ for $\gamma\in\Gamma$, then
$w=w'\gamma$.

A \emph{run} in $\cA$ is a sequence $(q_0,w_0),\ldots,(q_n,w_n)$ of
configurations and words $u_1,\ldots,u_n\in\Sigma^*$ such that
$(q_{i-1},w_{i-1})\autstep[u_i](q_{i},w_i)$ for $1\le i\le n$. Its
\emph{length} is $n$ and its \emph{initial} and \emph{final
  configuration} are $(q_0,w_0)$ and $(q_n,w_n)$, respectively. The
run is said to \emph{read} the word $u_1\cdots u_n$. The \emph{stack
  height} of the run is defined as $\max\{|w_i| \mid 0\le i\le n\}$.
We call the run \emph{positive} (resp. \emph{dually positive}) if $|w_i|\ge
|w_0|$ (resp. $|w_i|\ge 
|w_n|$) for every
$1\le i\le n$, i.e.  if the stack never drops below its initial
height (resp. is always above its final height).

We write $(q,w)\autsteps[u](q',w')$ for configurations $(q,w),(q',w')$
if there is a run with initial configuration $(q,w)$ and final
configuration $(q',w')$ that reads $u$. If there is such a run with
stack height $\le h$, then we write $(q,w)\autstepsh[u]{h}(q',w')$.
The \emph{language accepted by $\cA$} is
\[ \langof{\cA}=\{u\in\Sigma^* \mid  (q_0,\varepsilon)\autsteps[u](q_f,\varepsilon) \}. \]
There is also a language accepted with bounded stack height. For $h\in\nats$, we define
\[ \langof[h]{\cA}=\{u\in\Sigma^* \mid 
  (q_0,\varepsilon)\autstepsh[u]{h}(q_f,\varepsilon) \}. \] 
  Thus, $\langof[h]{\cA}$ is the set of words accepted by $\cA$ using runs
  of stack height at most $h$.
  
  In order
  to
exploit the symmetry between forward and backward computations in
pushdown automata, we will consider dual pushdown automata.  The
\emph{dual automaton} of $\cA$, denoted $\bar{\cA}$, is obtained from
$\cA$ by changing each edge $p\autstep[R|v]q$ into
$q\autstep[\rev{R}|\bar{v}] p$. Then
$\langof{\bar{\cA}}=\rev{\langof{\cA}}$. We will sometimes argue by
\emph{duality}, which is the principle that every statement that is
true for any $\cA$ is also true for any $\bar{\cA}$.

\paragraph{Petri nets} A \emph{(labeled) Petri net} is a tuple
$N=(\Sigma,S,T,\partial_0,\partial_1,\lambda,\mmap_0,\mmap_f)$ where $\Sigma$ is its \emph{input alphabet}, $S$ is a finite
set of \emph{places}, $T$ is a finite set of \emph{transitions},
$\partial_0,\partial_1\colon T\to\multiset{S}$ are maps that specify
an \emph{input marking} $\partial_0(t)$ and an \emph{output marking}
$\partial_1(t)$ for each transition $t\in T$,
$\lambda\colon T\to\Sigma\cup\{\varepsilon\}$ assigns labels to
transitions, and $\mmap_0,\mmap_f$ are its \emph{initial} and
\emph{final marking}. More generally, multisets
$\mmap\in \multiset{S}$ are called \emph{markings} of $N$.

For markings $\mmap_1,\mmap_2\in\multiset{S}$ and
$a\in\Sigma\cup\{\varepsilon\}$, we write $\mmap_1\autstep[a]\mmap_2$
if there is a transition $t\in T$ with $\lambda(t)=a$,
$\mmap_1\ge \partial_0(t)$, and
$\mmap_2=\mmap_1\ominus\partial_0(t)\oplus\partial_1(t)$. Moreover, we
write $\mmap_1\autsteps[w]\mmap_2$ if there are $n\in\nats$,
$a_1,\ldots,a_n\in\Sigma\cup\{\varepsilon\}$, and markings
$\mmap'_0,\ldots,\mmap'_n$ such that $w=a_1\cdots a_n$ and
$\mmap_1=\mmap'_0\autstep[a_1]\mmap'_1\autstep[a_2]\cdots\autstep[a_n]\mmap'_n=\mmap_2$.
Furthermore, we write $\mmap_1\autsteps\mmap_2$ if there exists a
$w\in\Sigma^*$ with $\mmap_1\autsteps[w]\mmap_2$.  The \emph{language
  accepted by $N$} is
$\langof{N}=\{w\in\Sigma^* \mid \mmap_0\autsteps[w]\mmap_f\}$.

For $k\in\nats$, a Petri net $N$ is \emph{$k$-bounded} if for every
marking $\mmap\in\multiset{S}$ with $\mmap_0\autsteps\mmap$, we have
$\card{\mmap}\le k$.

Our main results are as follows.

\begin{thm}[Succinct Downclosures for PDAs]\label{pda2bdd-pda}
  Given a pushdown automaton $\cA$, one can construct in polynomial time a
  pushdown automaton $\hat{\cA}$ so that
  $\dcl{\langof{\hat{\cA}}}=\dcl{\langof{\cA}}$. Moreover, we have
  $\langof{\hat{\cA}}=\langof[h]{\hat{\cA}}$ for some bound $h$ that is
  polynomial in the size of $\cA$.
\end{thm}

As a pushdown automaton with bounded stack can be simulated by a 1-bounded Petri net
(essentially, by keeping places for each position in the stack), we get the following
corollary and also the promised $\EXPSPACE$ upper bound.

\begin{cor}\label[cor]{pda2pn}
  Given a pushdown automaton $\cA$, one can construct in polynomial time
  a $1$-bounded Petri net $N$ with $\langof{N}=\dcl{\langof{\cA}}$.
\end{cor}

We now prove Theorem~\ref{pda2bdd-pda}.
The \emph{augmented automaton}
$\hat{\cA}=(Q,\Sigma,\hat{\Gamma},\hat{E},q_0,q_f)$ is defined as
follows. We first compute for any $p,q\in Q$ the subalphabet $\Delta_{p,q}(\cA)\subseteq\Sigma$ with 
\[ \Delta_{p,q}(\cA)=\{a\in\Sigma \mid \exists u\in M_{p,q}(\cA),~|u|_a\ge 1 \},\]
where $M_{p,q}(\cA)=\{ u\in\Sigma^*\mid \exists v\in\Gamma^*\colon
(p,\varepsilon)\autsteps[u](p,v),~ (q,v)\autsteps(q,\varepsilon)\}$.
Note that it is easy to construct in polynomial time a pushdown
automaton $\cA_{p,q}$ for the language $M_{p,q}(\cA)$: \\
$\cA_{p,q}$ has a set $ Q_p \cup Q_q$ of states consisting of
two disjoint copies of the states of $\cA$. The transitions between two states in $Q_p$ are inherited
from $\cA$ while for two states in $Q_q$ we replace the input on any
transition by $\varepsilon$ but retain the stack operations from $\cA$.
The start state is $p \in Q_p$ and the final state is $q \in Q_q$.
There is an $\varepsilon$-transition from $p \in Q_p$ to $q \in Q_q$. This concludes the construction of
$\cA_{p,q}$.

 Since
$a\in\Delta_{p,q}(\cA)$ iff
$M_{p,q}(\cA)\cap \Sigma^*a\Sigma^*\ne\emptyset$, we can decide in
polynomial time whether a given $a\in\Sigma$ belongs to
$\Delta_{p,q}(\cA)$ by checking the PDA for $M_{p,q}(\cA)\cap
\Sigma^*a\Sigma^*\ne\emptyset$ obtained by product
construction for emptiness. Thus, we can compute $\Delta_{p,q}(\cA)$
in
polynomial time.
We construct $\hat{\cA}$ as follows. For any $p,q\in Q$, we introduce
a fresh stack symbol $[p,q]$ and then we add edges
\begin{align}
  &p\autstep[\Delta_{p,q}(\cA)^*|[p,q]] p, &&q\autstep[\Delta_{q,p}(\bar{\cA})^*|\overline{[p,q]}] q. \label{newedges-appendix}
\end{align}
The following \lcnamecref{samedc} tells us that $\langof{\hat{\cA}}$ has the same downward closure as $\lang{\cA}$.

\begin{lem}\label[lem]{samedc}
$\langof{\cA}\subseteq\langof{\hat{\cA}}\subseteq\dcl{\langof{\cA}}$.
\end{lem}
\begin{proof}
Since the inclusion $\langof{\cA}\subseteq\langof{\hat{\cA}}$ is
obvious, we prove $\langof{\hat{\cA}}\subseteq\dcl{\langof{\cA}}$.  We
proceed by induction on the number $m$ of executions of new edges
(i.e. those from \cref{newedges-appendix}).
More specifically, we show that if $u\in\langof{\hat{\cA}}$ is
accepted using $m$ executions of new edges, then there is a
$u'\in\langof{\hat{\cA}}$ such that $u\subword u'$ and $u'$ is accepted using $<m$ executions of
new edges.

Suppose $u$ is accepted using a run $\rho$ with $m>0$ executions of
new edges.  Then $\rho$ must apply one edge
$p\autstep[\Delta_{p,q}(\cA)|[p,q]]p$ and thus also the edge
$q\autstep[\Delta_{q,p}(\bar{\cA})|\overline{[p,q]}] q$ to remove the
letter $[p,q]$ from the stack. Thus, $\rho$ decomposes as
$\rho=\rho_1\rho_2\rho_3\rho_4\rho_5$, where $\rho_2$ and $\rho_4$ are
the executions of the new edges. Let $u=u_1u_2u_3u_4u_5$ be the
corresponding decomposition of $u$.

The run $\rho_1$ must end in state $p$ and with some stack content
$w\in\Gamma^*$.  Then, $\rho_3$ is a run from $(p,w[p,q])$ to
$(q,w[p,q])$ and $\rho_5$ is a run from $(q,w)$ to $(q_f,\varepsilon)$
with $q_f\in F$.

Since $u_2$ and $u_4$ are read while executing the new edges, we have
$u_2\in\Delta_{p,q}(\cA)^*$ and $u_4\in\Delta_{q,p}(\bar{\cA})^*$. We
can therefore write $u_2=r_1\cdots r_k$ and $u_4=s_1\cdots s_\ell$
with $r_1,\ldots,r_k\in\Delta_{p,q}(\cA)$ and
$s_1,\ldots,s_\ell\in\Delta_{q,p}(\bar{\cA})$. By definition, this
means for each $1\le i\le k$, there is a word
$\tilde{r}_i\in M_{p,q}(\cA)$ that contains the letter
$r_i$. Likewise, for every $1\le i\le \ell$, there is a
$\tilde{s}_i\in M_{q,p}(\bar{\cA})$ that contains $s_i$.

Since $\tilde{r}_i\in M_{p,q}(\cA)$ and $\tilde{s}_j\in M_{q,p}(\bar{\cA})$
for $1\le i\le k$ and $1\le j\le \ell$, there are words $x_i$ and $y_j$ in $\Gamma^*$
so that
\begin{align*}
  &(p,\varepsilon)\autsteps[\tilde{r}_i] (p,x_i) &&\text{and} && (q,x_i)\autsteps (q,\varepsilon)\\
  &(p,\varepsilon)\autsteps (p,y_j) &&\text{and} && (q,y_i)\autsteps[\tilde{s}_j] (q,\varepsilon)
\end{align*}
We can therefore construct a new run $\rho'=\rho_1\rho'_2\rho'_3\rho'_4\rho_5$, where
\begin{align*}
  &\rho'_2: (p,w)\autsteps[\tilde{r}_1]\cdots\autsteps[\tilde{r}_k] (p,wx_1\cdots x_k) \autsteps \cdots \autsteps (p,wx_1\cdots x_ky_\ell\cdots y_1) \\
  &\rho'_4: (q,wx_1\cdots x_ky_\ell\cdots y_1) \autsteps[\tilde{s}_1]\cdots \autsteps[\tilde{s}_\ell](q,wx_1\cdots x_k)\autsteps\cdots\autsteps (q,w).
\end{align*}
Moreover, since $\rho_3$ is a positive run from $(p,w)$ to $(q,w)$, we obtain $\rho'_3$ from $\rho_3$
by replacing the prefix $w$ of every stack by $wx_1\cdots x_ky_1\cdots y_\ell$.

Then $\rho'$ reads some word
$u_1\tilde{r}_1\cdots \tilde{r}_kfu_3\tilde{s}_1\cdots\tilde{s}_\ell
gu_5$ for $f,g\in\Sigma^*$. Note that since $r_i$ occurs in
$\tilde{r}_i$ and $s_i$ occurs in $\tilde{s}_j$, we have
$u=u_1u_2u_3u_4u_5\subword u_1\tilde{r}_1\cdots
\tilde{r}_kfu_3\tilde{s}_1\cdots\tilde{s}_\ell gu_5$. \qedhere
\end{proof}

We now show that every word in $\hat{\cA}$ is accepted by a run with
bounded stack height. %

\begin{lem}\label[lem]{boundedstack}
  $\langof{\hat{\cA}}=\langof[h]{\hat{\cA}}$, where $h=2 |Q|^2$.
\end{lem}

Before we prove \cref{boundedstack}, we need another observation.
Just like \cref{samedc}, one can show that for any $p,q\in Q$, we have
$M_{p,q}(\cA)\subseteq M_{p,q}(\hat{\cA})\subseteq\dcl{M_{p,q}(\cA)}$
and in particular
\begin{align}
  &\Delta_{p,q}(\cA)=\Delta_{p,q}(\hat{\cA}) && \Delta_{q,p}(\bar{\cA})=\Delta_{q,p}(\bar{\hat{\cA}}), \label{samealphabets}
\end{align}
where the second identity follows from the first: Duality yields
$\Delta_{q,p}(\bar{\cA})=\Delta_{q,p}(\hat{\bar{\cA}})$ and since
$\hat{\bar{\cA}}$ and $\bar{\hat{\cA}}$ are isomorphic (i.e. they are
the same up to a renaming of stack symbols), we have
$\Delta_{q,p}(\hat{\bar{\cA}})=\Delta_{q,p}(\bar{\hat{\cA}})$.

We now prove \cref{boundedstack}.  Let $u\in\langof{\hat{\cA}}$. We
show that any minimal length accepting run $\rho$ reading $u$ must
have stack height $\le h$ and hence $u\in\langof[h]{\hat{\cA}}$.

Suppose the stack height of $\rho$ is larger than $h=2|Q|^2$. 
\paragraph{Claim:}$\rho$ decomposes into runs
$\rho_1,\rho_2,\rho_3,\rho_4,\rho_5$
reading $u_1,u_2,u_3,u_4,u_5$, respectively, so that there are
$p,q\in Q$ and $w\in\Gamma^*$ with
\begin{itemize}
\item $\rho_2$ is a positive run  from $(p,w)$ to $(p,wv)$ of length $\ge 2$
\item $\rho_3$ is a positive run from $(p,wv)$ to $(q,wv)$
\item $\rho_4$ is a dually positive run from $(q,wv)$ to $(q,w)$
\end{itemize}
Let $c_h$ be a configuration along $\rho$ with stack height at least 
$2|Q|^2+1$. Then there exist $|Q|^2$ configurations $c_2 \autsteps
c_4 \autsteps \cdots \autsteps c_{2|Q|^2} \autsteps c_h$
along $\rho$
such that $c_{2i}$ is the last time that the stack height is $2i$ before visiting $c_h$. 
Symmetrically, we have $c_h \autsteps c'_{2|Q|^2} \autsteps
\cdots c'_4 \autsteps c'_2$ where $c'_{2i}$ is the first occurrence
of stack height $2i$ after visiting $c_h$. Clearly by definition the run between
consecutive $c_{2i}$ (resp. $c'_{2i}$) configurations is positive 
(resp. dually positive). Additionally, the length of the run between them
must be at least two, because the stack heights differ by two. Considering the pair of states at each $c_{2i},c'_
{2i}$, there are $|Q|^2$ possibilities. Hence there must exist indices
$2i < 2j$ such that the $c_{2i}$ and $c_{2j}$ have the same state $p$
and
$c'_{2i}$ and $c'_{2j}$ have the same state $q$. It is now clear that
$\rho_2\colon c_{2i} \autsteps c_{2j}$, $\rho_3 \colon c_{2j} \autsteps
c'_{2j}$ and $\rho_4\colon c'_{2j} \autsteps c'_{2i}$ satisfy the
conditions of the claim.

These conditions imply that $u_2\in M_{p,q}(\hat{\cA})$ and
$u_4\in M_{q,p}(\bar{\hat{\cA}})$.  Therefore, we have
$u_2\in\Delta_{p,q}(\hat{\cA})^*=\Delta_{p,q}(\cA)^*$ and
$u_4\in \Delta_{q,p}(\bar{\hat{\cA}})^*=\Delta_{q,p}(\bar{\cA})^*$ by \cref{samealphabets}.

We obtain the run $\rho'$ from $\rho$ as follows. We replace $\rho_2$
by a single execution of the edge
$p\autstep[\Delta_{p,q}(\cA)^*|[p,q]]p$ reading $u_2$. Moreover, we
replace $\rho_4$ by a single execution of the edge
$q\autstep[\Delta_{q,p}(\bar{\cA})^*|\overline{[p,q]}] q$. Then
$\rho'$ is clearly a run reading $u=u_1u_2u_3u_4u_5$. Furthermore,
since $\rho_2$ has length $\ge 2$, but the single edge used instead in
$\rho'$ only incurs length $1$, $\rho'$ is strictly shorter than
$\rho$. This is in contradiction to the minimal length of $\rho$.
This completes the proof of Lemma~\ref{boundedstack} and thus also Theorem~\ref{pda2bdd-pda}.

\begin{remark}
  The augmented automaton $\hat{\cA}$ yields a very simple
  construction of a finite automaton (of exponential size) for
  $\dcl{\langof{\cA}}$. First, it is easy to construct a finite
  automaton for $\langof[h]{\hat{\cA}}$. Then, by introducing
  $\varepsilon$-edges, we get a finite automaton for
  $\dcl{\langof[h]{\hat{\cA}}}$, which, by \cref{samedc,boundedstack},
  equals $\dcl{\langof{\cA}}$.
\end{remark}

It is now straightforward to construct a polynomial size $1$-bounded
Petri net $N=(\Sigma,S,T,\partial_0,\partial_0,\mmap_0,\mmap_f)$ with
$\langof{N}=\langof[h]{\hat{\cA}}$, thus proving Corollary~\ref{pda2pn}. First, by adding states, we turn
$\hat{\cA}$ into a pushdown automaton
$\cA'=(Q',\Sigma,\hat{\Gamma},E',q_0,q_f)$, where every edge reads at
most one letter, i.e. every edge $p\autstep[R|v]q$ in $\cA'$ has
$R=\{x\}$ for some $x\in\Sigma\cup\{\varepsilon\}$ (this is done by
`pasting' the automaton for $R$ in place of the edge). Moreover, we
add
$\varepsilon$-edges, so that for every edge $p\autstep[\{x\}|v]q$, we
also have an edge $p\autstep[\{\varepsilon\}|v]q$. Then clearly
$\langof[h]{\cA'}=\dcl{\langof[h]{\hat{\cA}}}=\dcl{\langof{\cA}}$.

The net $N$ has a place $p$ for each state $p$ of $\cA'$ and for each
$i\in\{1,\ldots,h\}$ and $\gamma\in\hat{\Gamma}$, it has a place
$(i,\gamma)$. Moreover, for each $i\in\{0,\ldots,h\}$, it has a place
$s_i$.  Here, the idea is that a configuration $c=(p,\gamma_1\cdots \gamma_n)$
of $\cA'$ with $\gamma_1,\ldots,\gamma_n\in\hat{\Gamma}$ is represented as
a marking $\mmap_c=\multi{p,(1,\gamma_1),\ldots,(n,\gamma_n),s_n}$.  We call a
marking of this form a \emph{stack marking} and will argue that every
reachable marking in $N$ is a stack marking. The transitions in $N$
correspond to edges in $\cA'$. For each edge $p\autstep[\{x\}|v] q$
in $\hat{\cA}$, we add the following transitions:
\begin{itemize}
\item if $v=\bar{\gamma}$ for some $\gamma\in\hat{\Gamma}$, then we
  have for every $1\le n\le h$ a transition $t$ with
  $\partial_0(t)=\multi{p,(n,\gamma),s_n}$,
  $\partial_1(t)=\multi{q,s_{n-1}}$, and $\lambda(t)=x$.
\item if $v\in\hat{\Gamma}$, then for every $0\le n<h$, we add
  a transition $t$ with $\partial_0(t)=\multi{p,s_n}$,
  $\partial_1(t)=\multi{q,(n+1,v),s_{n+1}}$, and $\lambda(t)=x$.
\item if $v=\varepsilon$, then we add a transition $t$ with
  $\partial_0(t)=\multi{p}$, $\partial_1(t)=\multi{q}$, and $\lambda(t)=x$.
\end{itemize}
Then clearly every reachable marking is a stack marking and we have
$c\autstep[x]c'$ for configurations $c,c'$ of $\cA'$ of stack height
$\le h$ if and only if $\mmap_c\autstep[x]\mmap_{c'}$.  Therefore, if
we set $\mmap_0=\multi{q_0,s_0}$ and $\mmap_f=\multi{q_f,s_0}$ as
initial and final marking, we have
$\langof{N}=\langof[h]{\cA'}=\dcl{\langof{\cA}}$. 
This completes the proof of \cref{pda2pn}.

\bibliographystyle{alphaurl}
\bibliography{bibliography}

\appendix
\section{Compiling away internal actions}
\label{sec:internal-actions}

We have commented in the examples of Section~\ref{sec:preliminaries}
that internal actions and internal updates of the global state are useful
in modeling asynchronous programs from their programming language syntax.
Indeed, we note that the definition of asynchronous programs in~\cite{GantyM12} additionally uses a separate alphabet of \emph{internal actions},
in addition to the alphabet of handler posts.
We show how a model of asynchronous programs with internal actions
can be reduced to our, simpler, model.

Let $\cC$ be a language class over an alphabet $\Sigma$ of handler
names.
The definition of asynchronous programs with internal actions, as used by~\cite{SenV06,JhalaM07,GantyM12}, 
is as follows.\footnote{
  The language class $\cC$ in~\cite{GantyM12} is fixed to be the class of context
  free languages. Their definition generalizes to any language class.
}
An \textit{asynchronous program over $\cC$} with \emph{internal actions} (aka $\AP$ over $\cC$ with
internal actions) is a tuple
 \(\ap = (D,
\Sigma, \Sigmai,
\process,
R, d_0, \mmap_0),\) 
where $D$, $\Sigma$, $d_0$, $\mmap_0$ are as in Definition~\ref{defn:ap},
$\Sigmai$ is an alphabet of \textit{internal actions} disjoint from $\Sigma$,
the set $\process=(L_\sigma)_{\sigma \in \Sigma}$ consists of languages from
$\cC$ (one for each handler $\sigma \in \Sigma)$, and
 \(R=(D,\Sigma \cup \Sigmai,\delta)\) a deterministic finite state automaton where $D$ is the
set of states, $\Sigma \cup \Sigmai$ the alphabet and $\delta$ the
transition relation specifying the effect of each internal action on
the global state $D$.
We will write $d
\overset{w}{
\underset{R}
\Rightarrow^*}
d'$ to mean that there is a sequence of transitions with labels
$w_1w_2...w_n=w$ in the automaton $R$ using which we can reach $d'$
from $d$.

For alphabets $\Sigma,\Sigma'$ with $\Sigma \subseteq \Sigma'$, the {\em Parikh map}
$\parikh\colon\Sigma'^* \rightarrow \multiset{\Sigma}$ maps a word
$w\in \Sigma'^*$ to a multiset $\parikh(w)$ such that $\parikh(w)(a)$
is the number of occurrences of $a$ in $w$ for each $a \in \Sigma$.  For example,
$\parikh(abbab)(a)=2$, $\parikh(abbab)(b)=3$ and
\(\parikh(\varepsilon)=\multi{ }\).  For a language $L$, we define
$\parikh(L) = \set{\parikh(w)\mid w\in L}$.  If the alphabet $\Sigma$
is not clear from the context, we write $\parikh_\Sigma$ (usually $\Sigma'=\Sigma$).

The semantics of a $\ap$ is given as a labeled
transition system over
the set of configurations, with a transition relation $\rightarrow
\subseteq (D\times\multiset{\Sigma})\times \Sigma \times
(D\times\multiset{\Sigma})$ defined as follows: let \(\mmap,\mmap'\in\multiset{\Sigma}\), \(d,d'\in D\) and \(\sigma\in\Sigma\)
\begin{gather*}
  (d,\mmap\oplus\multi{\sigma})\underset{\ap}{
  \overset{\sigma}\rightarrow} (d',\mmap\oplus\mmap')\\
\text{ if{}f }\\
\exists w\in (\Sigma \cup \Sigmai)^* \colon d \overset{w}{
\underset{R}\Rightarrow^*}
 d' \land w \in L_{\sigma} \land \mmap'
= \parikh_{\Sigma}(w)\enspace. 
\end{gather*} 

We now show that internal actions can be compiled away.
Thus, for the algorithms presented in this paper, we use the more abstract,
language-theoretic version of Definition~\ref{defn:ap}, while we use
internal actions as syntactic sugar in examples.

\begin{lem}
\label{lem:internal-actions}
Let $\cC$ be a full trio.  Given an $\AP$ $\api$ with internal actions
over $\cC$, one can construct an $\AP$ $\ap$ over $\cC$ such that
both have identical sets of runs.
\end{lem}

\begin{proof}
The proof is along similar lines to that of Lemmas 4.3, 4.5 in~\cite{GantyM12}. 
Given \(\api = (D,
\Sigma, \Sigmai,
\process,
R, d_0, \mmap_0)\) we construct $\ap=(D,\Sigma,(L_c)_
{c\in\fC},d_0,\mmap_0)$ such that 
\begin{gather*}
(d,\mmap\oplus\multi{\sigma})\underset{\ap}{
  \overset{\sigma}\rightarrow} (d',\mmap\oplus\mmap')
\quad \text{ if{}f }  \quad
(d,\mmap\oplus\multi{\sigma})\underset{\api}{
  \overset{\sigma}\rightarrow} (d',\mmap\oplus\mmap') 
\end{gather*}
Let $L(R_{d,d'})$ be the language of the automaton $R$ when $d$ is the 
initial state and $d'$ is the accepting state. 
We define $L_{d \sigma d'}$ as:
\[ L_{d \sigma d'} := \proj_{\Sigma}(L_{\sigma} \cap L(R_{d,d'}))\]
First observe that 
the projection operation is a homomorphism and $L(R_{d,d'}))$ is a 
regular language; hence by virtue of $\cC$ being a full trio $L_{d 
\sigma d'}$ as defined above is in $\cC$.  The conditions $\exists 
w\in (\Sigma \cup \Sigmai)^* \colon d \overset{w}{
\underset{R}\Rightarrow^*}
 d' \land w \in L_{\sigma} \land \mmap' 
= \parikh_{\Sigma}(w)$ and $\exists w \in L_{d \sigma d'} \colon 
\parikh(w)=\mmap'$ are seen to be equivalent from the definition of 
$L_{d \sigma d'}$, concluding the 
proof of the lemma. 
\end{proof}

\section{Proof of Proposition~\ref{dcpreservation}}
\label{appendix:dcpreservation}
We prove the following proposition.

  Let $\ap=(D,\Sigma,(L_c)_{c\in\fC},d_0,\mmap_0)$ be an asynchronous
  program. Then $\dcl{\Runs{\dcl{\ap}}}=\dcl{\Runs{\ap}}$. In
  particular,
  \begin{enumerate}
  \item For every $d\in D$, $\ap$ can reach $d$ if and only if $\dcl{\ap}$ can reach $d$.
  \item $\ap$ is terminating if and only if $\dcl{\ap}$ is terminating.
  \item $\ap$ is bounded if and only if $\dcl{\ap}$ is bounded.
  \end{enumerate}

\begin{proof}
A run of the asynchronous program $\ap$ is defined as a sequence
$c_0,\sigma_1,c_1,\sigma_2,\ldots$ containing alternating
elements of configurations $c_i$ and letters $\sigma_i$ beginning
with the configuration $c_0=(d_0,\mmap_0)$. First we observe that
\begin{equation}
\label{eqn:runs_included_in_apdown}
  \Runs{\ap} \subseteq \Runs{\dcl{\ap}}
\end{equation}
 because every transition enabled in $\ap$ is also enabled in $
 \dcl{\ap}$. Next, we claim: 
 \begin{equation}
 \label{eqn:run_pdown_below_p}
  \forall \; \rho \in \Runs{\dcl{\ap}} \; \exists \; \rho' \in \Runs{\ap}
  \; \rho \leqruns \rho' 
 \end{equation}
 Let $\rho|_k=(d_0,\mmap_0),\sigma_1,
  (d_1,\mmap_1),\sigma_2,...,\sigma_k,(d_k,\mmap_k)$ be the 
  $2k+1-$length prefix of $\rho$. 
  We show that for each $\rho|_k$ there exists $\rho'_k \in \Runs{\ap}$
  such that $\rho|_k \leqruns \rho'_k$ and in addition, $\forall k \;
  \forall i 
  \leq k \; \rho'_k(i)=\rho'_{k+1}(i)$. We can then define $\rho'(i) :=
\rho'_i(i)$ and clearly $\rho \leqruns \rho'$. 

  We prove by induction on $k$.\\
  \underline{Base Case:} $\rho|_0=\rho'|_0=(d_0,\mmap_0)$.\\
  \underline{Induction Step:} Let $\rho|_k=(d_0,\mmap_0),\sigma_1,
  (d_1,\mmap_1),\sigma_2,...,(d_k,\mmap_k) \in \Runs{\dcl{\ap}}$. By 
  induction 
  hypothesis there is $\rho'_{k-1}=(d_0,\mmap_0),\sigma_1,
  (d_1,\mmap_1'),\sigma_2,...,(d_{k-1},\mmap_{k-1}') \in \Runs{\ap}$
  such that 
  $\rho_{k-1} \leqruns \rho'_{k-1}$. 

    \begin{gather*} 
     (d_{k-1},\mmap_{k-1}) \overset{\sigma_k}{
      \underset{\dcl{\ap}}\rightarrow} (d_k,\mmap_k)\\
      \Rightarrow \exists \mmap_{k-1}'' \colon \mmap_{k-1}=\mmap_{k-1}'' 
      \oplus 
      \multi{\sigma_k} \land (d_{k-1},\mmap_{k-1}''\oplus 
      \multi{\sigma_k}) 
      \overset{\sigma_k}{
      \underset{\dcl{\ap}}\rightarrow} (d_k,\mmap_k) \\
      \Rightarrow \exists w \in \Sigma^* \colon w \in \dcl{L_{d_{k-1}
      \sigma_k d_k}}
      \land \parikh(w) \oplus \mmap_{k-1}''=\mmap_k \\
      \Rightarrow \exists w' \in \Sigma^* \colon w \subword w' \land w' 
      \in L_{d_
      {k-1} \sigma_k d_k} \\    
      \Rightarrow (d_{k-1},\mmap_{k-1}) \overset{\sigma_k}{
      \underset{\ap}\rightarrow} (d_k,\mmap_k \oplus \mmap_\Delta) 
      \text{where }
      \mmap_\Delta \oplus \parikh(w)=\parikh(w') \\
      \Rightarrow (d_{k-1},\mmap_{k-1}') \overset{\sigma_k}{
      \underset{\ap}\rightarrow} (d_k,\mmap_k') \text{ where }
      \mmap_k'=\mmap_k \oplus \mmap_{\Delta} \oplus(\mmap'_{k-1} \ominus
      \mmap_{k-1})
    \end{gather*}
  Defining $\rho'_k := \rho'_{k-1},\sigma_k, (d_k,\mmap_k')$ we 
  see that $\rho|_k \leqruns \rho'_k$, completing the proof of Equation 
  \ref{eqn:run_pdown_below_p}. We are now ready to show that $
  \dcl{\Runs{
  \dcl{\ap}}}=\dcl{\Runs{\ap}}$. The direction $\dcl{\Runs{\ap}} \subseteq  \dcl{\Runs{\dcl{\ap}}}
  $ follows immediately from Equation 
  (\ref{eqn:runs_included_in_apdown}). Conversely, let 

  \begin{align*}
    & \rho=
  (s_0,\nmap_0),\sigma_1,(s_1,\nmap_1),\sigma_2,... \in \dcl{\Runs{
  \dcl{\ap}}} &
  \\
  \Rightarrow & \exists \rho' \in \Runs{\dcl{\ap}} \; \rho \leqruns 
  \rho'& \\
  \Rightarrow & \exists \rho'' \in \Runs{\ap} \; \rho' \leqruns \rho'' 
  \; &\text{ by Equation (\ref{eqn:run_pdown_below_p})} \\
  \Rightarrow & \rho \in \dcl{\Runs{\ap}} &
  \end{align*}
We have proved that $\dcl{\Runs{\dcl{\ap}}}=\dcl{\Runs{\ap}}$. We now 
show that each of the three properties i.e. safety, termination and 
boundedness only depend on the downclosure of the runs. 

  \noindent\underline{Safety:} \begin{align*}
    &d \text{ is reachable in } \ap\\
    \text{ iff } &\exists \rho=(d_0,\mmap_0),\sigma_1,
  (d_1,\mmap_1),\sigma_2,...,\sigma_k,(d_k,\mmap_k) \in \Runs{\ap}
  \colon d_k=d 
  \end{align*}
  By Equation \ref{eqn:runs_included_in_apdown}, we know $\rho \in 
  \Runs{\ap}$ implies $\rho \in \dcl{\Runs{\ap}}$. Conversely, if there 
  is $\rho'=(s_0,\nmap_0),\sigma_1,
  (s_1,\nmap_1),\sigma_2,...,\sigma_k,(s_k,\nmap_k) \in \dcl{\Runs 
  {\ap}}$ with $s_k=d$,
  then by Equation \ref{eqn:run_pdown_below_p}, there is $\rho=
  (d_0,\mmap_0),\sigma_1,
  (d_1,\mmap_1),\sigma_1,...,\sigma_k,(d_k,\mmap_k) \in \Runs{\ap}$
  with $\rho'\leqruns 
  \rho$ which implies $d_k=d$. Hence we have 
   \begin{align*}
    &d \text{ is reachable in } \ap\\
    \text{ iff }&
    \exists \rho=(s_0,\nmap_0),\sigma_1,
  (s_1,\nmap_1),\sigma_2,...,\sigma_k,(s_k,\nmap_k) \in \dcl{\Runs 
  {\ap}} \colon s_k=d 
  \end{align*}
  By a similar argument as above we also have:\\
  \underline{Termination:}
\begin{align*} 
  & \ap \text{ does not terminate } \\
  \text{ iff } & \exists \rho \in 
  \Runs{\ap}
  \colon \rho \text { is infinite } \\
  \text{ iff } & \exists \rho \in 
  \dcl{\Runs{\ap}}
  \colon \rho \text{ is infinite }
  \end{align*}
  \underline{Boundedness:}

    \begin{align*} 
  & \ap \text{ is bounded } \\
  \text{ iff } & \exists N \in \nats \; \forall \rho \in 
  \Runs{\ap} \; \forall i \; |\rho(2i).m| < N  \\
  \text{ iff } & \exists N \in \nats \; \forall \rho \in 
  \dcl{\Runs{\ap}} \; \forall i \; |\rho(2i).m| < N 
  \end{align*}
  In each of the three cases, the property only depends on the downclosure and hence one may equivalently replace $\ap$ by $\dcl{\ap}$
  since $\dcl{\Runs{\dcl{\ap}}}=\dcl{\Runs{\ap}}$. 
    \end{proof}

\section{Proof of Proposition \ref{prop:pathslangs_to_word_schemes}}
\label{sec:appendix_ho}

We begin with a simple observation.
  For every $\hors$ $\rs=(\Sigma,\nonterm, S, \rewrite)$, there exists another $\hors$ 
$\rs'=(\Sigma',\nonterm', S, \rewrite')$ where $\Sigma'=\{\br,\e\}
\cup \tilde{\Sigma}$ with $\br$ of arity 2, $\e$ of arity 0
and all symbols in $\tilde{\Sigma}$ of arity 1; such that $\langpath
(\rs)=\langpath(\rs')$.

By rewriting every terminal symbol $A$ of arity $n \geq 2$ using the rule  \[A
x_1 x_2\cdots x_n \twoheadrightarrow \br(x_1 \br (x_2 \br (\cdots\br(x_{n-1} x_n))\]
and by replacing the occurrence of every nullary symbol in $\rs$ by $\ltr{e}$, we get $\rs'$, which has the same path language. 

\noindent
  We assume due to the above observation 
  that
  $\Sigma=\{\br,\e\}
\cup \tilde{\Sigma}$ where $\br$ is binary, $\e$ is nullary and all
letters in $\tilde{\Sigma}$ are unary. Define $\rs'=(\Sigma',\nonterm',
S', \rewrite')$ as $\Sigma'=\tilde{\Sigma} \cup \{ \e \},
\nonterm'=\nonterm
\cup \{B:\otype \rightarrow (\otype \rightarrow \otype)\}, \rewrite'=
\{r[\br / B] \mid r \in \rewrite \} \cup \{
Bxy \twoheadrightarrow x, Bxy \twoheadrightarrow y \}$, where by $r[\br / B]$ we
mean the rule $r$ with $\br$ uniformly replaced by $B$. Note that
the only new non-terminal symbol introduced is $B$, which is of
order 1. Hence the obtained word scheme $\rs'$ is of the same order as
$\rs$. \\
\underline{$\langpath(\rs) \subseteq \langword(\rs')$:} Let $w \in
\langpath(\rs)$. Therefore there exists a finite path $\Pi$ in a
sentential form $t$ such that $\forall \bar{n} \in dom(\Pi) \; w_\Pi(\bar{n}) \in \Sigma$. We derive $t':=
t[\br/B]$ using the
corresponding rules in $\rewrite'$. Note that the corresponding path
$\Pi'$ in $t'$ satisfies $\forall \bar{n} \in
\dom(\Pi') w_{\Pi'}(\bar{n}) \in \Sigma \cup \{B \}$. We then apply either $Bxy
\twoheadrightarrow x$ or $Bxy \twoheadrightarrow y$ to each $B$ in $t'$ according to
the path $\Pi'$ to obtain $w$ in the word scheme.\\
\underline{$\langword(\rs') \subseteq \langpath(\rs)$:}
We define the order $\prefix$ on sequences of natural numbers $
\bar{n},\bar{m}$ as $\bar{n} \prefix \bar{m}$ if $\bar{m}=\bar{n}
\bar{k}$ for some sequence $\bar{k}$. 

 Suppose that
for given a sentential form $t'$ of $\rs'$ there exists a sentential
form $t$ of $\rs$ and a map $\alpha:\dom(t') \rightarrow \dom(t)$ 
(simply called \textit{embedding} henceforth)
satisfying the following conditions:
\begin{itemize}
  \item $\forall \bar{n} \in \dom(t')$
  \begin{equation} \label{eqn:map_first_cond}
   t'(\bar{n}) = \begin{cases}
     B \text{ if } t(\alpha(\bar{n}))= \br \\
 t(\alpha(\bar{n})) \text{ otherwise}.
  \end{cases}\end{equation}

  \item $\forall \bar{n}, \bar{m} \in \dom(t'), \bar{n} \prefix \bar{m}
  $ implies
   \begin{equation}
   \label{eqn:map_second_cond}
        \alpha(\bar{n}) \prefix \alpha(\bar{m})\text{, and}
      \end{equation}    

      \begin{equation}
\label{eqn:map_third_cond}
\begin{aligned}
  & (\forall \bar{l} \; (\bar{n} \sprefix \bar{l} \sprefix \bar{m}) 
  \text{ implies } t'
  (\bar{l}) \notin \tilde{\Sigma}) \text{ implies } \\ 
  &(\forall
  \bar{k} \; \alpha(\bar{n}) \sprefix \bar{k} \sprefix \alpha(\bar{m})
  \text{ implies } t(\bar{k}) \notin \tilde{\Sigma} )
\end{aligned}
      \end{equation}
\end{itemize}
Informally, Equations 
\ref{eqn:map_first_cond}, \ref{eqn:map_second_cond} and
\ref{eqn:map_third_cond} state
the following: $\alpha$ preserves labels, except for the case of $\br$
where it maps to $B$, $\alpha$ preserves the order $\prefix$ and the
images of any two nodes with node labels from $\tilde{\Sigma}$ with no
node label from $\tilde{\Sigma}$ in between are mapped to two such
nodes with the same
property.

We will show by induction on the length of the derivation that such a
pair $(t,\alpha)$ always exists given some $t'$.
Let us see how the existence of such $t$ and $\alpha$ gives us the
proposition.
Consider a word $w=w_1w_2...w_n \in \langword(\rs')$. In other words,
there is a term $t'=w_1(w_2...(w_n(\e))...)$ such that $S' 
\underset{\rs'}\twoheadrightarrow^* t'$. Corresponding to this, we have a
sentential form $S
\underset{\rs}\twoheadrightarrow^* t$ and $\alpha$ satisfying the given
conditions. In particular, there exists a path $\Pi$ with $\dom(\Pi)
\subseteq dom(t)$ which is the path in $t$ connecting the $\alpha
(\varepsilon)$ to $\alpha(0^n)$. It is immediate that $\proj_
{\tilde{\Sigma}}
(\Pi)=w$. It remains to show the existence of $t$ and $\alpha$ by
induction on the length of derivations.\\

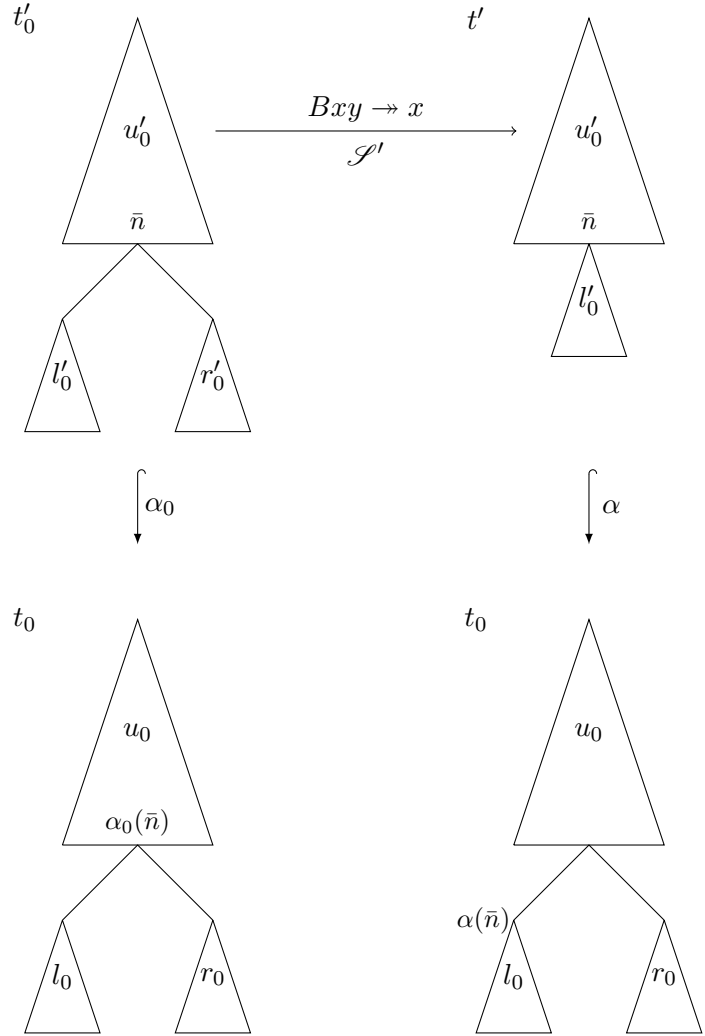
\begin{figure}
\begin{tikzpicture}
  \def \twidth {2cm} %
  \def \theight {3cm} %
  \def \ruleshift {6cm} %
  \def \embedshift {-8cm} %

  \draw (0,0) -- (\twidth,0) -- (\twidth/2,\theight) -- (0,0);
  \node at (-0.5,\theight) {$t'_0$};

  \node at (\twidth/2,0.3) {\small $\bar{n}$};
  \node at (\twidth/2,\theight/2) {$u_0'$};

  \coordinate (l) at ($(\twidth/2,-1)+(-1,0)$);
   \coordinate (r) at ($(\twidth/2,-1)+(1,0)$);
  \draw (r) -- (\twidth/2,0) -- (l);
  \draw (l) -- ($(l)+(-\twidth/4,-\theight/2)$) -- ($(l)+
  (\twidth/4,-\theight/2)$) --(l);
  \node at ($(l)-(0,\theight/4)$) {$l_0'$};

\draw (r) -- ($(r)+(-\twidth/4,-\theight/2)$) -- ($(r)+
  (\twidth/4,-\theight/2)$) --(r);
  \node at ($(r)-(0,\theight/4)$) {$r_0'$};

\coordinate (ab) at (\twidth+1,\theight/2);
\coordinate (ae) at ($(ab)+(4,0)$);
\draw[->] (ab) --(ae);
 \node at ($0.5*(ab)+0.5*(ae)+(0,0.3)$) {$Bxy \twoheadrightarrow x$};
 \node at ($0.5*(ab)+0.5*(ae)+(0,-0.3)$) {$\rs'$};

\begin{scope}[xshift=\ruleshift]
  \draw (0,0) -- (\twidth,0) -- (\twidth/2,\theight) -- (0,0);
  \node at (-0.5,\theight) {$t'$};

  \node at (\twidth/2,0.3) {\small $\bar{n}$};
  \node at (\twidth/2,\theight/2) {$u_0'$};

  \coordinate (l2) at ($(\twidth/2,0)$);
  \draw (l2) -- ($(l2)+(-\twidth/4,-\theight/2)$) -- ($(l2)+
  (\twidth/4,-\theight/2)$) --(l2);
  \node at ($(l2)-(0,\theight/4)$) {$l_0'$};
\end{scope}

\begin{scope}[yshift=\embedshift]
  \draw (0,0) -- (\twidth,0) -- (\twidth/2,\theight) -- (0,0);
  \node at (-0.5,\theight) {$t_0$};

  \node at (\twidth/2,0.3) {\small $\alpha_0(\bar{n})$};
  \node at (\twidth/2,\theight/2) {$u_0$};

  \coordinate (l) at ($(\twidth/2,-1)+(-1,0)$);
   \coordinate (r) at ($(\twidth/2,-1)+(1,0)$);
  \draw (r) -- (\twidth/2,0) -- (l);
  \draw (l) -- ($(l)+(-\twidth/4,-\theight/2)$) -- ($(l)+
  (\twidth/4,-\theight/2)$) --(l);
  \node at ($(l)-(0,\theight/4)$) {$l_0$};

\draw (r) -- ($(r)+(-\twidth/4,-\theight/2)$) -- ($(r)+
  (\twidth/4,-\theight/2)$) --(r);
  \node at ($(r)-(0,\theight/4)$) {$r_0$};
\end{scope}

\coordinate (a1) at (\twidth/2,-3);
\coordinate (a2) at (\twidth/2,-4);
\draw [right hook-latex] (a1) -- (a2); 
\node at ($0.5*(a1)+0.5*(a2)+(0.3,0)$) {$\alpha_0$};

\begin{scope}[yshift=\embedshift, xshift=\ruleshift]
  \draw (0,0) -- (\twidth,0) -- (\twidth/2,\theight) -- (0,0);
  \node at (-0.5,\theight) {$t_0$};

  \node at (\twidth/2,0.3) {};
  \node at (\twidth/2,\theight/2) {$u_0$};

  \coordinate (l) at ($(\twidth/2,-1)+(-1,0)$);
   \coordinate (r) at ($(\twidth/2,-1)+(1,0)$);

   \node at ($(l)+(-0.4,0)$) {\small $\alpha(\bar{n})$};

  \draw (r) -- (\twidth/2,0) -- (l);
  \draw (l) -- ($(l)+(-\twidth/4,-\theight/2)$) -- ($(l)+
  (\twidth/4,-\theight/2)$) --(l);
  \node at ($(l)-(0,\theight/4)$) {$l_0$};

\draw (r) -- ($(r)+(-\twidth/4,-\theight/2)$) -- ($(r)+
  (\twidth/4,-\theight/2)$) --(r);
  \node at ($(r)-(0,\theight/4)$) {$r_0$};
\end{scope}

\begin{scope}[xshift=\ruleshift]
\coordinate (a3) at (\twidth/2,-3);
\coordinate (a4) at (\twidth/2,-4);
\draw [right hook-latex] (a3) -- (a4); 
\node at ($0.5*(a3)+0.5*(a4)+(0.3,0)$) {$\alpha$};
\end{scope}
\end{tikzpicture}
\caption{Construction of embedding $\alpha$ from $\alpha_0$ in the
case of the rule $Bxy \twoheadrightarrow x$}
\label{fig:alpha_Brule}
\end{figure}

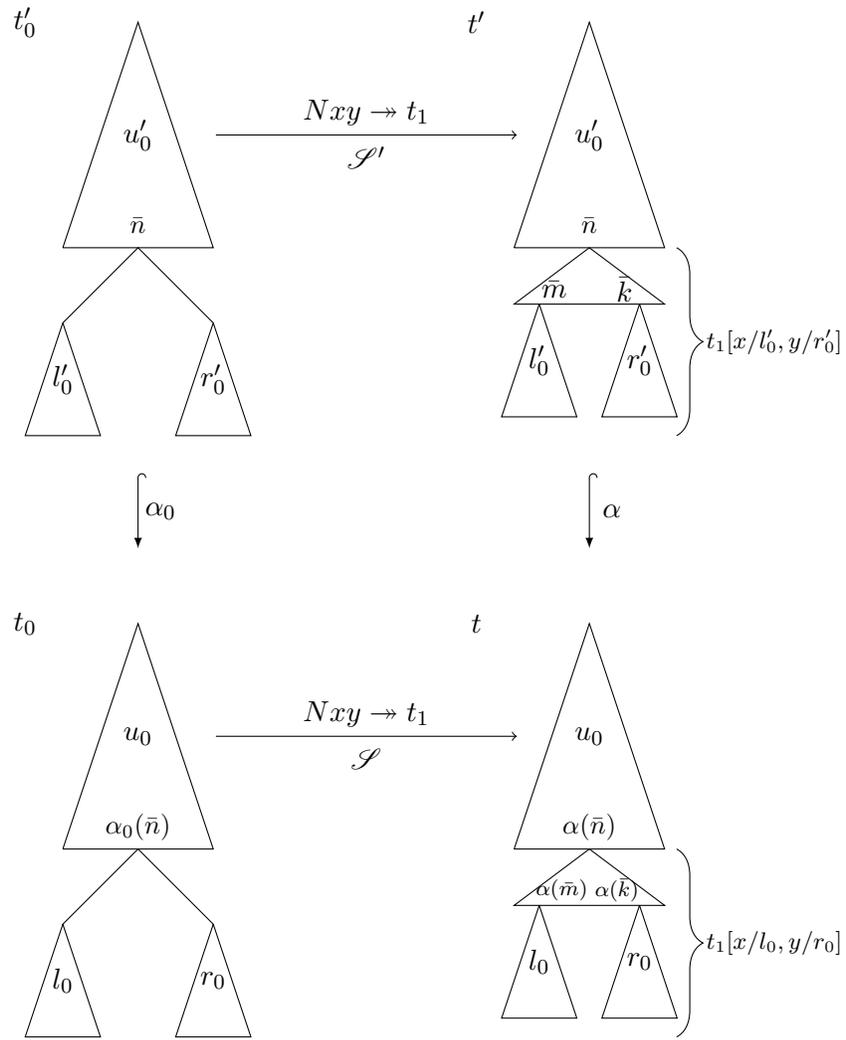
\begin{figure}
\begin{tikzpicture}
  \def \twidth {2cm} %
  \def \theight {3cm} %
  \def \ruleshift {6cm} %
  \def \embedshift {-8cm} %

  \draw (0,0) -- (\twidth,0) -- (\twidth/2,\theight) -- (0,0);
  \node at (-0.5,\theight) {$t'_0$};

  \node at (\twidth/2,0.3) {\small $\bar{n}$};
  \node at (\twidth/2,\theight/2) {$u_0'$};

  \coordinate (l) at ($(\twidth/2,-1)+(-1,0)$);
   \coordinate (r) at ($(\twidth/2,-1)+(1,0)$);
  \draw (r) -- (\twidth/2,0) -- (l);
  \draw (l) -- ($(l)+(-\twidth/4,-\theight/2)$) -- ($(l)+
  (\twidth/4,-\theight/2)$) --(l);
  \node at ($(l)-(0,\theight/4)$) {$l_0'$};

\draw (r) -- ($(r)+(-\twidth/4,-\theight/2)$) -- ($(r)+
  (\twidth/4,-\theight/2)$) --(r);
  \node at ($(r)-(0,\theight/4)$) {$r_0'$};

\coordinate (ab) at (\twidth+1,\theight/2);
\coordinate (ae) at ($(ab)+(4,0)$);
\draw[->] (ab) --(ae);
 \node at ($0.5*(ab)+0.5*(ae)+(0,0.3)$) {$Nxy \twoheadrightarrow t_1$};
 \node at ($0.5*(ab)+0.5*(ae)+(0,-0.3)$) {$\rs'$};

\begin{scope}[xshift=\ruleshift]
  \draw (0,0) -- (\twidth,0) -- (\twidth/2,\theight) -- (0,0);
  \node at (-0.5,\theight) {$t'$};

  \node at (\twidth/2,0.3) {\small $\bar{n}$};
  \node at (\twidth/2,\theight/2) {$u_0'$};

  \coordinate (l2) at ($(\twidth/2,0)$);
  \draw (l2) -- ($(l2)+(-\twidth/2,-\theight/4)$) -- ($(l2)+
  (\twidth/2,-\theight/4)$) --(l2);

  \coordinate (l3) at ($(l2)+(-\twidth/3,-\theight/4)$);
  \node at ($(l3)+(0.2,0.2)$) {$\bar{m}$};
  \coordinate (r3) at ($(l2)+(\twidth/3,-\theight/4)$);
  \node at ($(r3)+(-0.2,0.2)$) {$\bar{k}$};
  \draw (l3) -- ($(l3)+(-\twidth/4,-\theight/2)$) -- ($(l3)+
  (\twidth/4,-\theight/2)$) --(l3);
  \node at ($(l3)-(0,\theight/4)$) {$l_0'$};

  \draw (r3) -- ($(r3)+(-\twidth/4,-\theight/2)$) -- ($(r3)+
  (\twidth/4,-\theight/2)$) --(r3);
  \node at ($(r3)-(0,\theight/4)$) {$r_0'$};

  \draw [decorate,decoration=
  {brace,amplitude=10pt},xshift=-4pt,yshift=0pt] ($(\twidth/2,0)+
  (1.3,0)$)
  -- ($(\twidth/2,-2.5)+(1.3,0)$) node [black,midway,xshift=1.3cm] 
  {\footnotesize
  $t_1[x/l'_0,y/r'_0]$};
\end{scope}

\begin{scope}[yshift=\embedshift]
  \draw (0,0) -- (\twidth,0) -- (\twidth/2,\theight) -- (0,0);
  \node at (-0.5,\theight) {$t_0$};

  \node at (\twidth/2,0.3) {\small $\alpha_0(\bar{n})$};
  \node at (\twidth/2,\theight/2) {$u_0$};

  \coordinate (l) at ($(\twidth/2,-1)+(-1,0)$);
   \coordinate (r) at ($(\twidth/2,-1)+(1,0)$);
  \draw (r) -- (\twidth/2,0) -- (l);
  \draw (l) -- ($(l)+(-\twidth/4,-\theight/2)$) -- ($(l)+
  (\twidth/4,-\theight/2)$) --(l);
  \node at ($(l)-(0,\theight/4)$) {$l_0$};

\draw (r) -- ($(r)+(-\twidth/4,-\theight/2)$) -- ($(r)+
  (\twidth/4,-\theight/2)$) --(r);
  \node at ($(r)-(0,\theight/4)$) {$r_0$};
\end{scope}

\coordinate (a1) at (\twidth/2,-3);
\coordinate (a2) at (\twidth/2,-4);
\draw [right hook-latex] (a1) -- (a2); 
\node at ($0.5*(a1)+0.5*(a2)+(0.3,0)$) {$\alpha_0$};

\begin{scope}[yshift=\embedshift]
\coordinate (ab) at (\twidth+1,\theight/2);
\coordinate (ae) at ($(ab)+(4,0)$);
\draw[->] (ab) --(ae);
 \node at ($0.5*(ab)+0.5*(ae)+(0,0.3)$) {$Nxy \twoheadrightarrow t_1$};
  \node at ($0.5*(ab)+0.5*(ae)+(0,-0.3)$) {$\rs$};

\end{scope}

\begin{scope}[yshift=\embedshift, xshift=\ruleshift]
  \draw (0,0) -- (\twidth,0) -- (\twidth/2,\theight) -- (0,0);
  \node at (-0.5,\theight) {$t$};

  \node at (\twidth/2,0.3) {\small $\alpha(\bar{n})$};
  \node at (\twidth/2,\theight/2) {$u_0$};

  \coordinate (l2) at ($(\twidth/2,0)$);
  \draw (l2) -- ($(l2)+(-\twidth/2,-\theight/4)$) -- ($(l2)+
  (\twidth/2,-\theight/4)$) --(l2);

  \coordinate (l3) at ($(l2)+(-\twidth/3,-\theight/4)$);
    \node at ($(l3)+(0.3,0.2)$) {\scriptsize $\alpha(\bar{m})$};
  \coordinate (r3) at ($(l2)+(\twidth/3,-\theight/4)$);
    \node at ($(r3)+(-0.3,0.2)$) {\scriptsize $\alpha(\bar{k})$};
  \draw (l3) -- ($(l3)+(-\twidth/4,-\theight/2)$) -- ($(l3)+
  (\twidth/4,-\theight/2)$) --(l3);
  \node at ($(l3)-(0,\theight/4)$) {$l_0$};

  \draw (r3) -- ($(r3)+(-\twidth/4,-\theight/2)$) -- ($(r3)+
  (\twidth/4,-\theight/2)$) --(r3);
  \node at ($(r3)-(0,\theight/4)$) {$r_0$};

  \draw [decorate,decoration=
  {brace,amplitude=10pt},xshift=-4pt,yshift=0pt] ($(\twidth/2,0)+
  (1.3,0)$)
  -- ($(\twidth/2,-2.5)+(1.3,0)$) node [black,midway,xshift=1.3cm] 
  {\footnotesize
  $t_1[x/l_0,y/r_0]$};
\end{scope}

\begin{scope}[xshift=\ruleshift]
\coordinate (a3) at (\twidth/2,-3);
\coordinate (a4) at (\twidth/2,-4);
\draw [right hook-latex] (a3) -- (a4); 
\node at ($0.5*(a3)+0.5*(a4)+(0.3,0)$) {$\alpha$};
\end{scope}
\end{tikzpicture}
\caption{Construction of embedding $\alpha$ from $\alpha_0$ when
the rule is $Nxy \twoheadrightarrow t_1$}
\label{fig:alpha_Nrule}
\end{figure}

\noindent\underline{Base case:} This is trivial since $t=t'=S$.\\
\underline{Induction step:} Suppose $t_0' \underset{r
\in \rewrite'}\twoheadrightarrow
t'$ where $t_0'$ is a sentential form. By induction hypothesis, there
is a sentential form $t_0$ of $\rs$ and $t_0'$ embeds into $t_0$ via
map $\alpha_0$. Assume that the rule $r$ is applied at position $
\bar{n}$ on $t_0'$. We now have two cases to consider:\\
\underline{Case 1:} The rule $r$ is $Bxy \underset{\rs'}\twoheadrightarrow x$
(the case $Bxy \underset{\rs'}\twoheadrightarrow y$ being symmetric). By
induction hypothesis, we have $t_0$ and and an embedding $\alpha_0$ of
$t_0'$ into $t_0$. Refering to Figure \ref{fig:alpha_Brule}, we see
that $t$ can be taken to be $t_0$ and $\alpha$ maps all nodes in the
subtree $u'_0$ into $u_0$ as before, while the subtree $l'_0$ rooted
at $\bar{n}$ is mapped into $l_0$. It is immediate that $\alpha$
preserves the order and since by induction hypothesis $\alpha_0(\bar
{n})$ has label $\br$, Equation \ref{eqn:map_third_cond} is also
satisfied by $\alpha$ since no new $\tilde{\Sigma}$ labelled nodes
have been added. \newline
\underline{Case 2:} We demonstrate for the case when a non-terminal
of arity two is replaced for an easier reading of the proof: the rule
$r$ is $N x y
\underset{\rs'}\twoheadrightarrow t_1$ for some nonterminal $N \neq B$.
Refering to Figure \ref{fig:alpha_Nrule}, the rule replaces the
subtree rooted at $\bar{n}$ in $t'_0$ with $t_1[x/l'_0,y/r'_0]$ where
$l'_0,r'_0$ are respectively the left and right subtrees of the
subtree rooted at $\bar{n}$ in $u'_0$. For nodes in the subtree $u'_0$
of $t'$, $\alpha$ mimics $\alpha_0$. By induction we know that the
subtrees 
$l'_0,r'_0$ of $t'_0$ embed respectively into $l_0,r_0$ of $t_0$. Thus
$\alpha$ maps every subtree $l'_0$ (resp. $r'_0$) rooted at $ \bar{m}$
(resp. $\bar{k}$) in $t_1[x/l'_0,y/r'_0]$ into the corresponding
subtree $l_0$ (resp. $r_0$) rooted at $\alpha(\bar{m})$ (resp. $\alpha
(\bar{k})$). Note that $\alpha(\bar{m})$ can be defined as the root of the $i$-th occurrence of $l_0$ from left to right if $\bar{m}$ is the $i$-th occurrence of $l'_0$ from left to right. Nodes in $t_1[x/l'_0,y/r'_0]$ which are not in any of the
$l_0'$ (or $r_0'$) subtrees (i.e.) between $\bar{n}$ and $\bar{m}$ (or
$\bar{k})$) have corresponding nodes in  $t_1[x/l_0,y/r_0]$ to which
they can be mapped.  Label preservation and
order
preservation immediately follow by appeal to the induction hypothesis.
In order to see that Equation \ref{eqn:map_third_cond} holds, consider
consecutive $\tilde{\Sigma}$ nodes $\bar{n}', \bar{m}'$ in $t'$ (i.e.
$\bar{n}' <_{\pre} \bar{m}'$ and there are no $\tilde{\Sigma}$ labels
in between). If both $\bar{n}'$ and $\bar{m}'$ are in $u_0'$ or in one
of the $l'_0$ (or $r_0'$) then the induction hypothesis applies. In the
case where $\bar{n}' \in u_0'$ and $\bar{m}' \in l_0'$ (resp. $r_0'$,)
this means
that no $\tilde{\Sigma}$ labels are present in the path from $\bar{n}$
to $\bar{m}$ (resp. $\bar{k}$), $\bar{n}'$ to $\bar{n}$ nor $\bar{m}$ 
(resp. $\bar{k}$) to $\bar{m}'$. The path from $\alpha(\bar{n})$ to
$\alpha(\bar{m})$ is identical to that from $\bar{n}$ to $\bar{m}$ (resp. $\bar{k}$).
The induction hypothesis also implies that the first $\ltr{a}$
label for some $\ltr{s} \in  \tilde{\Sigma}$ in $l_0'$ must be mapped
to first $\ltr{a}$ label in $l_0$ (or $u_0'$ does not contain any $
\tilde{\Sigma}$ labels). Hence there are no $\tilde{\Sigma}$ labels
between $\alpha(\bar{m})$ (resp. $\alpha(\bar{k})$) and $\alpha(
\bar{m}')$ and similarly between
$\alpha(\bar{n}')$ and $\alpha(\bar{n})$. The final case is when
either $\bar{n}'$ or $\bar{m}'$ lies between $\alpha(\bar{n})$ and
$\alpha(\bar{m})$ (resp. $\alpha(\bar{k})$). There are subcases here to
consider when the other point lies in $u_0'$, $l_0$ or also between
$\bar{n}$ and $\bar{m}$. In all of these subcases, it easily follows
that
there are no extra $\tilde{\Sigma}$ labels introduced in between two
consecutive nodes which are in the image of $\alpha$.

\end{document}